\documentclass[11pt,a4paper]{article}
\usepackage[linesnumbered,ruled]{algorithm2e}
\usepackage{titlesec}
\SetAlCapNameFnt{\small }
\SetAlCapFnt{\small }
\usepackage[title]{appendix}
\usepackage{soul}
\usepackage{amsmath,amssymb,amsbsy,amsthm,calc,cases,enumerate,url,framed,fancyhdr,lipsum,tikz,enumitem,natbib,etoolbox,bbm}
\usepackage{graphicx}
\usepackage[pageanchor,hypertexnames=false]{hyperref}
\usepackage{minitoc}
\usepackage{multirow} 
\usepackage{tabularx}

\newcommand{\unparskip}{\vspace{-\parskip}}

\titlespacing\section{0pt}{12pt plus 4pt minus 2pt}{0pt plus 2pt minus 2pt}
\titlespacing\subsection{0pt}{12pt plus 4pt minus 2pt}{0pt plus 2pt minus 2pt}
\titlespacing\subsubsection{0pt}{12pt plus 4pt minus 2pt}{0pt plus 2pt minus 2pt}

\newtheorem{theorem}{Theorem}
\newtheorem{corollary}{Corollary}
\newtheorem{lemma}{Lemma}%
\usepackage{bbm}
\usepackage{stmaryrd}
\newtheorem{proposition}{Proposition}%
\newtheorem{assumption}{Assumption}%

\BeforeBeginEnvironment{proof}{\unparskip}
\BeforeBeginEnvironment{proof_outline}{\unparskip}
 \AfterEndEnvironment{proof}{\unparskip}
  \AfterEndEnvironment{proof_outline}{\unparskip}
 
\BeforeBeginEnvironment{theorem}{\vspace{0.5\parskip}}
\BeforeBeginEnvironment{corollary}{\vspace{0.5\parskip}}
\BeforeBeginEnvironment{lemma}{\vspace{0.5\parskip}}
\BeforeBeginEnvironment{proposition}{\vspace{0.5\parskip}}

\BeforeBeginEnvironment{thm}{\vspace{0.5\parskip}}
\BeforeBeginEnvironment{cor}{\vspace{0.5\parskip}}
\BeforeBeginEnvironment{prop}{\vspace{0.5\parskip}}
\BeforeBeginEnvironment{lem}{\vspace{0.5\parskip}}

\newcommand{\iid}{\overset{\mathrm{iid}}{\sim}}

\newcommand{\dotcup}

\usepackage{xr}

\theoremstyle{remark}
\newtheorem{example}{Example}%
\theoremstyle{remark}
\newtheorem{remark}{Remark}%

\usepackage[normalem]{ulem}

\begin{document}

\title{On the design distribution for predictive Bayesian regression}

\author{Wanyue Sun$^{1}$ \ and Edwin Fong$^{1}$\thanks{Corresponding author. Email: chefong@hku.hk}  \\ \\
$^1$Department of Statistics and Actuarial Science, University of Hong Kong}
\date{}

\maketitle
\vspace{-5mm}

\begin{abstract}
The predictive approach to Bayesian inference accesses the posterior distribution via a sequence of one-step-ahead predictives, enabling inference via predictive resampling without Markov chain Monte Carlo. In the random-design regression setting, an explicit specification of the predictive design distribution is required, yet the impact of this choice has received little formal attention. We study the role of this predictive design distribution in parametric martingale posteriors for regression, and identify predictive notions of identifiability and design invariance that are essential for valid inference, particularly in the high-dimensional regression setting. Building on these foundations, we introduce a novel class of parametric martingale posteriors for regression that satisfies a weak form of these desiderata, and naturally accommodates the high-dimensional setting through regularization. We then illustrate our method through a simulation.
\end{abstract}

\section{Introduction}

\subsection{Parametric martingale posteriors for regression}\label{sec:pmp}
The predictive approach specifies the Bayesian model directly through a sequence of one-step-ahead predictives \citep{fortini2020quasi, berti2023bayesian, fong2023martingale}. A main advantage is that posterior sampling is free of Markov chain Monte Carlo, instead relying on a sequential imputation scheme known as predictive resampling. One key example is the parametric martingale posterior \citep{holmes2023statistical,garelli2024asymptotics, fortini2025exchangeability, fong2026asymptotics}, which utilizes parametric predictives and stochastic gradient descent. A related approach is that of \citet{battaglia2026variational} where a variational posterior predictive is utilized.

We consider random-design parametric regression under the data generating process \(Y \mid X \sim P_{\beta^*}(\cdot \mid X)\) and $X \sim P_X^*$, from which we observe data $(X_{1:n},Y_{1:n})$. Predictive resampling in this setup involves first imputing future covariates $(X_{i})_{i \geq n+1}$ independently and identically distributed from a chosen predictive design distribution $P_X$, which is usually the empirical distribution $\mathbb{P}_X$ of $X_{1:n}$. Next, we impute future responses $(Y_i \mid X_i) \sim P_{i-1}(\cdot \mid X_i)$ for $i \geq n+1$, where $P_{i-1}(y \mid x)$ is the conditional predictive distribution given $(X_{1:i-1},Y_{1:i-1})$. The distribution of the limit of an estimator $\beta_i = \beta(X_{1:i},Y_{1:i})$ as $i \to \infty$ is then termed the martingale posterior. In this paper, we formalize the validity of this predictive resampling procedure for parametric regression, which in practice must be truncated at a finite $i = N$. 

A key property of traditional Bayes is that the posterior of $\beta$ does not depend on the design distribution under standard assumptions \citep{gelman2013bayesian}.  However, the predictive framework explicitly requires the specification of $P_X$, and the effect of this choice has received little attention in the literature. In this paper, we formally investigate the role of $P_X$ during predictive resampling, and establish the importance of predictive notions of identifiability and design invariance, which is particularly relevant when $X$ is high-dimensional. Building on this, we introduce a parametric martingale posterior for regression which, unlike previous works, satisfies a weak version of the aforementioned desiderata. Our method naturally accommodates high-dimensional regression through prior regularization. We then demonstrate our method in a simulation study.

\subsection{Setup}
\label{sec_setup}

Let $P_X$ denote the predictive design distribution, which does not contain unknown parameters, and may not equal $P_X^*$. Next, let $\left\{P_\beta(\cdot \mid x): x \in \mathbb{R}^p, \beta \in \mathbb{R}^p\right\}$ denote a probability kernel with density $p_\beta(y \mid x)$ for $y \in \mathcal{Y} \subseteq \mathbb{R}$. Except in Assumption \ref{as:ident} and Proposition \ref{prop:doob_regression}, we specialize to the one-parameter model  $P_\beta(\cdot \mid x) = P_Y(\cdot \mid \beta^\top x)$, where  $\{P_Y(\cdot \mid t) : t \in \mathbb{R}\}$ is a probability kernel with density $p_Y(y \mid t)$.
We assume that the score function $r(y \mid t) :=  ({\partial}/{\partial t}) \log p_Y(y \mid t)$ exists and satisfies $\int r(y \mid t) \, P_Y(dy \mid t) = 0$, and define
$$s(\beta, y \mid x) := \nabla_\beta \log p_Y(y \mid \beta^\top x) = x \, r\left( y \mid \beta^\top x \right).$$
 Next, we write the conditional second moment of $r(y \mid t)$ as
$$w(t) = \int r(y \mid t)^2 \, P_Y(dy \mid t) $$
where we assume $0 < w(t) < \infty$ for all $t \in \mathbb{R}$.
This framework encompasses a wide range of common models, including generalized linear models with known dispersion and Student-$t$ regression with fixed scale. 
For later discussions involving martingales, we write the natural filtration as $\mathcal{F}_i = \sigma \{(X_{n+1}, Y_{n+1}), \dots, (X_i, Y_i)\}$ for $i \geq n+1$, where we write $\mathcal{F}_n = \{\emptyset, \Omega\}$ to highlight dependence on the fixed observations $(X_{1:n},Y_{1:n})$.

\section{Predictive identifiability and design invariance}
\subsection{Conjugate example}
In this subsection, we empirically demonstrate the role of the predictive design distribution $P_X$ under predictive resampling when $p_i(y \mid x)$ is chosen to be the Bayesian posterior predictive density before formalizing our observations with Doob's theorem in the next subsection. We demonstrate two key properties satisfied by the Bayesian posterior predictive:
\begin{enumerate}
    \item \textit{Identifiability}: If $\beta$ is not identifiable under $P_X$, then predictive resampling is invalid.
    \item \textit{Design invariance}: If $\beta$ is identifiable under $P_X$, then the posterior on $\beta$ is invariant to the choice of $P_X$. 
\end{enumerate}
While the first point is natural, it unfortunately precludes the use of the obvious choice of $\mathbb{P}_X$ when $p > n$ or there is strong collinearity, as we will formalize shortly. The second point is desirable from a modelling perspective, which we outline in Remark \ref{rem:design_invariance}, but achieving this for the general martingale posterior is nontrivial and forms the study of Section \ref{sec:pmp}. \vspace{2mm}

\begin{example}
\label{ex1}
Let $p_{\beta}(y\mid x) = \mathcal{N}(y; \,\beta^\top x, 1)$ and $\pi(\beta) = \mathcal{N}(\beta; \, 0, {I_{p}})$. Given $\left(Y_{1:i}, X_{1:i}\right)$, the posterior under standard settings is conjugate of the form $\pi(\beta\mid X_{1:i}, Y_{1:i}) = \mathcal{N}(\beta; \, {\beta}_{i}, \Sigma_{i})$, where $\beta_i$ and $\Sigma_i$ have well-known forms. The posterior predictive density is then
\begin{align*}
    p_{i}(y \mid x)= 
        \mathcal N\!\left(y;\, \beta_{i}^\top x, 1 + x^\top\,\Sigma_{i}\,x\right).
\end{align*}
We now simulate a dataset of size $n = 100, p = 400$ from the data generating process $X \sim P^*_X = {0.5\,\mathcal{N}(-1,\,I_{p}) + 0.5\,\mathcal{N}(1,\,I_{p})}$ and $(Y \mid X) \sim P_{\beta^*}(\cdot \mid X)$. Crucially, this places us in the high-dimensional regime as $p > n$. 
\begin{figure}
\begin{center}
\includegraphics[width=0.95\textwidth]{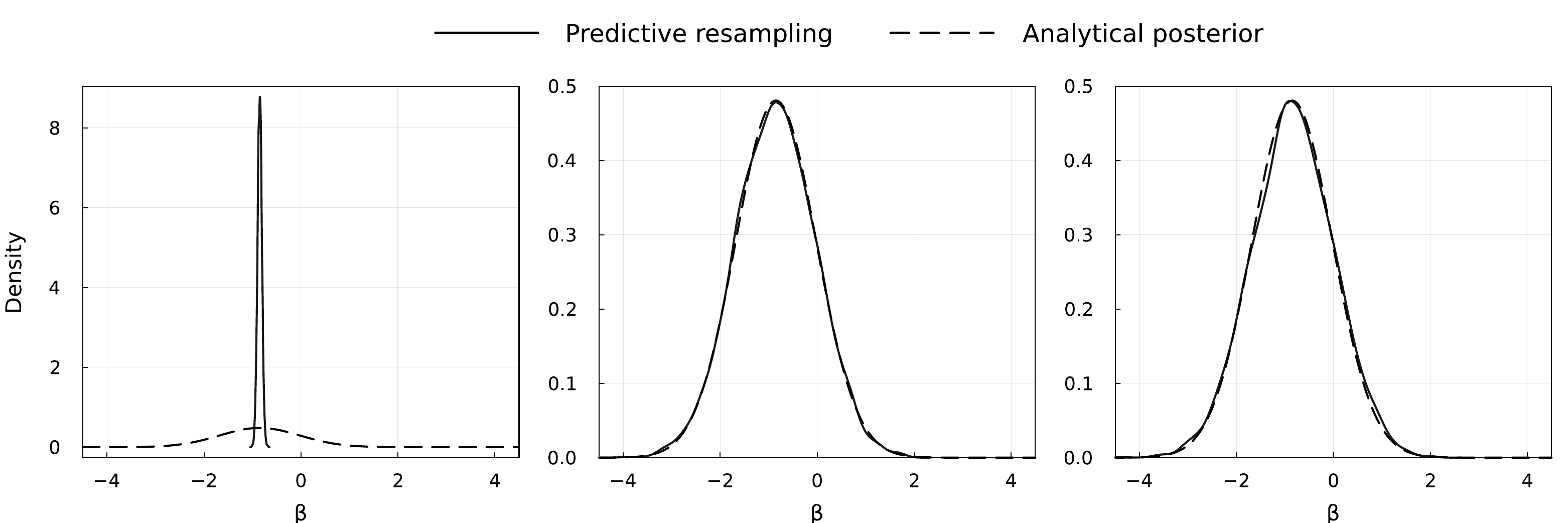}
\end{center}
\vspace{-5mm}
\caption{Posteriors for one component of $\beta$ through predictive resampling with (Left) $P_X = \mathbb{P}_X$; (Middle) $P_X = P^*_X$;  (Right) $P_X = \mathcal N(0, I_{p})$. The analytic posterior is plotted as a baseline.}
\label{fig:pr_vs_bayes_sim_lm}
\end{figure}
Figure \ref{fig:pr_vs_bayes_sim_lm} compares predictive resampling with the analytical posterior under three different choices of $P_X$, where we use the posterior mean $\beta_i$ as the estimator. 
First, Figure \ref{fig:pr_vs_bayes_sim_lm} (Left) shows that when $P_X = \mathbb{P}_X$, predictive resampling deviates from the truth.  As discussed in Section \ref{sec:non_ident} of the Appendix, predictive resampling can still converge in this setting, so this discrepancy is alarming. The underlying reason for this mismatch is that $P_X$ does not satisfy an identifiability condition, as the covariance matrix of $P_X = \mathbb{P}_X$ is rank-deficient when $p > n$. Second, we see that once $P_X$ is identifiable, as in Figures \ref{fig:pr_vs_bayes_sim_lm} (Middle) and (Right), predictive resampling not only recovers the correct posterior, but is also invariant to the choice of \(P_X\). A key point is that $P_X = \mathcal{N}(0,I_p)$, chosen purely for convenience, avoids the need to model $P^*_X$.
\end{example}

\subsection{Doob's theorem}
The preceding simulation example illustrates that the limiting behavior of predictive resampling depends critically on whether the regression parameter is identifiable under the chosen covariate distribution $P_X$. 
Proposition~\ref{prop:doob_regression} formalizes this connection by specializing Doob's consistency theorem \citep{Doob1949} to the random-design regression setting introduced in Section~\ref{sec_setup}, thereby making concrete the notions of predictive identifiability and design invariance. We begin with the case where $P_\beta(\cdot \mid x)$ is a general conditional model. 

\begin{assumption}\label{as:ident}
    For any $\beta \neq \beta'$, there exists a measurable set $B \subseteq \mathbb{R}^p$, possibly depending on $(\beta,\beta')$, such that $P_X(B) > 0$ and $P_\beta(\cdot\mid x)\neq P_{\beta'}(\cdot\mid x)$ for all  $x\in B$.
\end{assumption}
\begin{proposition}
\label{prop:doob_regression}
Let $\Pi$ be a prior distribution on $\beta \in \mathbb{R}^p$. For $i \geq n$, let $\Pi(\cdot \mid \mathcal F_i)$ denote the posterior distribution of the form
$\Pi(d\beta \mid \mathcal{F}_i) \propto \Pi(d\beta)\prod_{j = 1}^{i}p_{\beta}(Y_j \mid X_j)$, where $(X_{1:n},Y_{1:n})$ are fixed observations.  Consider predictive resampling as follows for $i\ge n+1$:
\[
X_{i} \sim P_X, \quad 
(Y_i\mid X_i) \sim P_{i-1}(\cdot\mid X_i), \quad P_{i-1}(\cdot \mid x)
=
\int P_\beta(\cdot \mid x)\,\Pi(d\beta\mid\mathcal F_{i-1}).
\]
Let $\beta_i = E(\beta\mid\mathcal F_i)$ and $E\left(\|\beta\|\mid \mathcal{F}_n\right) < \infty$.
 Then $(\beta_i, \mathcal F_i)_{i\ge n+1}$ is a uniformly integrable martingale and
$\beta_i \to \beta_\infty$
almost surely. Furthermore, under Assumption \ref{as:ident}, we have $\beta_\infty \sim \Pi(\cdot\mid\mathcal F_n)$.
\end{proposition}
Proposition \ref{prop:doob_regression} guarantees that under the identifiability assumption, predictive resampling is equivalent to posterior sampling. We now specialize to the setting where $P_{\beta}(\cdot \mid x) = P_Y(\cdot \mid \beta^\top x)$, which forms the focus of the rest of the paper. Next, we present a simple sufficient condition on the predictive design distribution \(P_X\) that ensures this identifiability requirement holds.\vspace{2mm}
\begin{assumption}\label{as:suff_ident}
    Let $P_X$ be such that $E_{P_X}(XX^\top)$ is finite and positive definite. 
\end{assumption}\vspace{-3mm}
\begin{proposition}
\label{prop:identifiability_sufficient}
Let $P_{\beta}(\cdot \mid x) = P_Y(\cdot \mid \beta^\top x)$, where the one-parameter model is identifiable, that is $P_Y(\cdot\mid t) \neq P_Y(\cdot\mid t')$ for all $t\neq t'$. 
Then Assumption \ref{as:suff_ident} implies Assumption \ref{as:ident}. 
\end{proposition}\vspace{2mm}

\begin{remark}
Let $P_{\beta}(\cdot \mid x) = P_Y(\cdot \mid \beta^\top x)$, and suppose the support of $P_X$ lies within a proper subspace $S\subsetneq\mathbb{R}^p$. For any nonzero $w\in S^\perp$, we have $(\beta+w)^\top x=\beta^\top x$ for all $x\in S$, implying $P_{\beta}(\cdot \mid x) = P_{\beta + w}(\cdot \mid x)$ for $P_X$-almost all $x$. Since $w \neq 0$, this violates Assumption \ref{as:ident}. This pathology arises immediately for the empirical distribution $\mathbb{P}_X$ when $p > n$, as the rank deficiency of $X_{1:n}$ yields a non-zero $w \in \mathrm{Null}(X_{1:n})$ such that $w^\top X_i=0$ for all $i=1,\ldots, n$. Hence $\mathbb{P}_X$ cannot guarantee identifiability when $p>n$, consistent with the discrepancy observed in Figure \ref{fig:pr_vs_bayes_sim_lm} (Left). 
\end{remark}\vspace{2mm}

\begin{remark}
\label{rem:design_invariance}
If Assumption \ref{as:ident} holds, then Proposition \ref{prop:doob_regression} implies that when using the Bayesian posterior predictive for predictive resampling, the limiting distribution of $\beta_i$ is invariant to the choice of $P_X$ within the identifiable class. This is shown in Figures \ref{fig:pr_vs_bayes_sim_lm} (Middle) and (Right). We term this property \textit{design invariance} and turn our focus to this property in the next section. 
\end{remark}

\section{Parametric martingale posterior}
\label{sec:pmp}
\subsection{Design invariance}\label{sec:pmp_des}
Design invariance is desirable in the predictive Bayesian framework, as it enables substituting \(P^*_X\) by a convenient working model such as \( P_X = \mathcal N( 0, I_p)\), without affecting the posterior for $\beta$. This is particularly useful in the high-dimensional setting where exact modelling of $P^*_X$ is often infeasible. While this property holds for traditional Bayesian predictives, it is unfortunately not attained for general parametric martingale posteriors proposed in the literature. As an example, \citet{holmes2023statistical,fortini2025exchangeability} introduce an update rule based on stochastic gradient descent with a plug-in predictive, which for the univariate case of Example~\ref{ex1} is
\begin{equation}
\label{eq:var_update}
 X_{i}\sim \mathcal N(0,\kappa^2), \quad (Y_i\mid X_i) \sim \mathcal{N}(\beta_{i-1} X_{i}, 1), \quad \beta_i = \beta_{i-1} + i^{-1}\,X_{i}\,\left(Y_{i}-\beta_{i-1} X_{i}\right).
\end{equation}
Under \eqref{eq:var_update}, the posterior variance can be computed as  $\operatorname{Var}(\beta_\infty \mid \mathcal F_n) =
\kappa^2\sum_{i=n+1}^{\infty} i^{-2}
\approx
\kappa^2\,n^{-1}$.
In this case, doubling $\kappa^2$ would result in a proportional inflation of the posterior variance for $\beta$, despite the conditional predictive $p_{i}(y \mid x)$ being the same, making posterior uncertainty of coefficients sensitive to arbitrary rescaling of the design distribution. We also highlight that the method of \citet{fong2026asymptotics}, which preconditions the score with the observed Fisher information computed from $X_{1:n}$, is still not design invariant, and is furthermore not well-defined if $p > n$. 

\subsection{Recursive update rule}
We now introduce a novel class of parametric martingale posteriors that attains a weak form of design invariance, and is suitable for the $p > n$ setting. As usual, we consider the plug-in predictive $P_i(\cdot \mid x) = P_{Y}(\cdot  \mid \beta_i^\top x)$, where we will only consider identifiable one-parameter models as motivated by Proposition \ref{prop:identifiability_sufficient}. {We highlight that the role of identifiability is more subtle for the martingale posterior compared to traditional Bayes; we discuss this in Section \ref{sec:pmp_ident} of the Appendix.} We initialize $\beta_n$ with an estimate from the real observations $(X_{1:n},Y_{1:n})$, and introduce the following novel recursive updating rule for $i \geq n+1$:
\begin{align}
\label{eq:recur_update}
\beta_i
=
\beta_{i-1}
&+
\alpha_{i-1}(X_i)\,
\mathcal{I}_i^{-1}\,
s(\beta_{i-1}, Y_i \mid X_i)
\end{align}
where $s(\beta,y \mid x)$ is the score function, $\mathcal{I}_i$ is a Fisher information term, and $\alpha_i(x)$ is a scalar correction factor. The latter two are given by
\begin{align*}
\mathcal{I}_i
=
\sum_{j=1}^i w\left( \beta_n^\top X_j \right) X_j X_j^\top
+
D(\beta_n), \quad \alpha_{i}(x)
=
\left[
\frac{
w\left(\beta_n^\top x\right)
}{
w\left( \beta_{i}^\top x\right)
}\bigl\{1 + w\left(\beta_n^\top x\right)\, x^\top \mathcal{I}_{i}^{-1} x\bigr\}
\right]^{1/2}
\end{align*}
and $D(\beta_n)$ is a chosen diagonal regularization matrix with non-negative entries to ensure invertibility of $\mathcal{I}_i$ for all $i \geq n$. 
While the Fisher information term is familiar (now computed from $X_{1:i}$), the scalar $\alpha_i(x)$ and diagonal matrix $D(\beta_n)$ are new. We provide further motivation for (\ref{eq:recur_update}) through connections to Bayesian ridge regression in Section \ref{sec:update_motiv} of the Appendix. 

The scalar correction factor $\alpha_{i}(x)$ is reminiscent of the posterior predictive variance in Example \ref{ex1}. This is no coincidence, as its sole role is to induce a weak form of design invariance by adaptively scaling the size of the update in (\ref{eq:recur_update}) to reduce the impact of the choice of $P_X$. We explore this notion of weak design invariance in the next subsection. We further note that $\alpha_{i}(x)$ is an easily computable scalar, does not require knowledge of $P_X$, and only depends on $(\mathcal{F}_{i},x)$ so it does not affect the martingale. 

The second new ingredient is the diagonal matrix $D(\beta_n)$, which acts as prior shrinkage to ensure positive definiteness of the preconditioning Fisher information term $\mathcal{I}_i$ in the $p > n$ regime. We will let the regularization depend on the initial estimate $\beta_n$, and we consider specific choices of $D(\beta)$ in Section \ref{sec:plugin} corresponding to different shrinkage priors.

\subsection{Weak design invariance and asymptotic normality}
In our setting, we refer to the distribution of $\beta_\infty$ when it exists as the martingale posterior. For general predictive schemes, it is not a simple task to ensure the martingale posterior is invariant to $P_X$. However, it is possible to obtain a weaker condition: we define \textit{weak design invariance} as the property where the posterior mean and covariance matrix of $\beta_\infty$ are invariant to the choice of $P_X$. This is attained by our update (\ref{eq:recur_update}), which we formalize below.
\begin{theorem}[Weak design invariance]
\label{thm:predictive_invariance_mp}
 Under predictive resampling with update rule (\ref{eq:recur_update}) and $P_i(\cdot \mid x) = P_{Y}(\cdot  \mid \beta_i^\top x)$, $(\beta_i, \mathcal{F}_i)_{i \geq n+1}$ is a martingale bounded in $L^2$ and $\beta_i \to \beta_\infty$ almost surely, so $E(\beta_\infty \mid \mathcal{F}_n)$ does not depend on $P_X$.
 If Assumption \ref{as:suff_ident} holds and $w(t)$ is continuous and strictly positive for all $t \in \mathbb{R}$, then $\textnormal{Cov}(\beta_\infty \mid \mathcal{F}_n)$ also does not depend on $P_X$. 
\end{theorem}
Invariance of the posterior mean follows directly from the uniformly integrable martingale, as we have $E(\beta_\infty \mid \mathcal{F}_n)= \beta_n$. Interestingly, the scalar correction term $\alpha_{i}(x)$ is sufficient for design invariance for the full posterior covariance matrix, which follows from a Sherman-Morrison formula argument. 
We highlight that weak design invariance may fail without the identifiability condition given by Assumption \ref{as:suff_ident}; see Section \ref{sec:non_ident} in the Appendix for more details. 

An obvious question arises: how weak is weak design invariance? In Section \ref{sec:sim} and Section \ref{sec:addit_examples} of the Appendix, we show in practice that the above martingale posterior is indistinguishable under different choices of $P_X$.  Furthermore, we will now show that the martingale posterior is asymptotically normal, so the posterior mean and covariance matrix asymptotically characterize the full posterior.

We first outline the necessary assumptions for asymptotic normality. Our first assumption is regarding convergence of the initial estimates $\beta_n$ and $\mathcal{I}_n$ computed from the real observations $(X_{1:n},Y_{1:n})$, now treated as independent and identically distributed samples from $P^*$. We will then consider a sequence of martingale posteriors centered at initial estimates $\beta_n$ and take $n\to \infty$.  

\begin{assumption}\label{as:In_conv}
Let the following hold $P^{*\infty}$-almost surely: $\mathcal{I}_n$ is positive definite for all $n$, $\beta_n$ converges, and $n^{-1}\,\mathcal{I}_n$ converges to a positive definite matrix.
\end{assumption}

Next, we offer two choices of conditions on the predictive design distribution $P_X$ depending on the form of $w(t)$ and $k(t)$, where  $k(t) = {{\int r(y \mid t)^4 \,P_Y(dy \mid t) }}/{w(t)^2}$ is the conditional kurtosis.\vspace{2mm}

\renewcommand{\theassumption}{\arabic{assumption}(i)}
\begin{assumption}\label{as:lyapunov_easy}
Suppose $w(t) = w_0$ and  
$k(t) =k_0$ where $w_0$ and $k_0$ are positive finite constants.   Let $P_X$ satisfy $E_{P_X}\left(\|X\|^4 \right) < \infty$.
\end{assumption}
\addtocounter{assumption}{-1}
\renewcommand{\theassumption}{\arabic{assumption}(ii)}
\begin{assumption}\label{as:lyapunov_hard}
Suppose $w(t)$ and $k(t)$ are continuous and strictly positive for all $t \in \mathbb{R}$. Let $P_X$ have compact support.
\end{assumption}
\renewcommand{\theassumption}{\thesection.\arabic{assumption}}
The forms of $w(t)$ and $k(t)$ in Assumption \ref{as:lyapunov_easy} 
hold for location families like the Gaussian, Laplace, and Student-$t$ families, and also the fixed-shape Gamma family with a log link. The more general case in Assumption \ref{as:lyapunov_hard} then handles models like logistic and Poisson regression. We provide further discussion on the assumptions in Section \ref{sec:assumptions} of the Appendix.

\begin{theorem}[Asymptotic normality]\label{thm:as_invar}
Suppose Assumptions \ref{as:suff_ident}, \ref{as:In_conv} and either \ref{as:lyapunov_easy} or \ref{as:lyapunov_hard} hold. Let $\beta_{n\infty}$ denote the limit under predictive resampling initialized at $\beta_n$ as in Theorem \ref{thm:predictive_invariance_mp}. Then the posterior law of $n^{1/2}\,(\beta_{n\infty} - \beta_n)$ converges weakly $P^{*\infty}$-almost surely to a Gaussian as $n \to \infty$.
\end{theorem}

We highlight that the technical machinery required in Theorem \ref{thm:as_invar} is noticeably heavier than in \cite{fong2026asymptotics}, primarily due to the more complex terms $\alpha_i(x)$ and $\mathcal{I}_i$. In particular, we must work with a weaker version of the Lindeberg condition and Doob's maximal inequalities, where these techniques may also be useful for future results. Finally, we highlight that analogous versions of Theorems \ref{thm:predictive_invariance_mp} and \ref{thm:as_invar} hold in the fixed-design case where $P_X$ is replaced with a deterministic design sequence $(X_{i})_{i \geq 1}$. We elaborate on this in Section \ref{sec:fixed_design} of the Appendix.

\subsection{Initial plug-in estimate}
\label{sec:plugin}
To initialize the recursive updates of \eqref{eq:recur_update}, we require a plug-in estimate $\beta_n$ and the regularization matrix $D(\beta_n)$ which ensures invertibility of $\mathcal{I}_i$. In the low-dimensional setting, one could set $\beta_n$ to the maximum likelihood estimate and $D(\beta_n) = 0$. In the high-dimensional setting, sparsity is typically assumed, and regularization is incorporated via shrinkage priors. In this case, we will let $\beta_n$ be the maximum a posteriori estimate computed from $(X_{1:n},Y_{1:n})$ for the chosen prior, and $D(\beta_n)$ will be allowed to depend on the initial estimate $\beta_n$.

Under a scale mixture of Gaussians prior, $D(\beta)$ can be motivated to be a diagonal precision matrix. We summarize the forms of $D(\beta)$ for common priors in Table~\ref{tab:D_n_summary}, with derivations deferred to Section \ref{sec:D_motiv} of the Appendix. For the Bayesian LASSO \citep{park2008bayesian} and continuous spike-and-slab \citep{george1993variable} priors, the diagonal entries are local shrinkage estimates computed from $\beta$. 
In the Gaussian continuous spike-and-slab case, $\gamma^{(j)}$ is the posterior inclusion probability of the $j$-th component $\beta^{(j)}$ and takes the form
\begin{align*}\label{eq:css}
\gamma^{(j)} = \frac{\theta\,\mathcal{N}(\beta^{(j)};\,  0, v_1)}{\theta\,\mathcal{N}(\beta^{(j)}; \,  0, v_1) 
      + (1-\theta)\,\mathcal{N}(\beta^{(j)}; \, 0, v_0)}
\end{align*}
where $\theta \in (0,1)$ is the prior inclusion probability, and $v_0$ and $v_1$ are the spike and slab variances respectively with $v_0 \ll v_1$.

\begin{table}[h]
\centering
\caption{Examples of $D(\beta)$ for common shrinkage priors; $\beta^{(j)}$ is the $j$-th component of $\beta$. }
\vspace{3mm}
\begin{small}
{\begin{tabular}{ccc}
Shrinkage prior & Diagonal matrix $D(\beta)$ & Description \\[4pt]
Ridge                  & $\tau I_p$          & $\tau =$ global precision \\[2pt]
LASSO            & $\text{Diag}(\tau^{(1)}, \dots, \tau^{(p)})$ & $\tau^{(j)} = \lambda\,{|\beta^{(j)}|^{-1}}$ \\[2pt]
Spike-and-slab         & $\text{Diag}(\tau^{(1)}, \dots, \tau^{(p)})$ & $\tau^{(j)} = (1-\gamma^{(j)})\,v_0^{-1} + \gamma^{(j)}\,v_1^{-1}$
\end{tabular}}
\end{small}
\label{tab:D_n_summary}
\end{table}

\subsection{Practical considerations for $P_X$}
\label{sec:px}
While weak design invariance ensures robustness of the posterior to the predictive design distribution $P_X$, its choice remains necessary for numerical implementation. A simple default that  satisfies Assumption \ref{as:suff_ident} and is straightforward to sample from is $P_X = \mathcal{N}(0, \kappa^2 I_p)$ for some $\kappa>0$. Another surprising benefit of design invariance is the freedom to specify $P_X$ based on computational convenience, without affecting posterior uncertainty. 
Predictive resampling requires truncation at sufficiently large $N$ such that $\beta_N$ has little remaining uncertainty. By increasing $\kappa$, we can increase the expected Fisher information under $P_X$, corresponding to more information about $\beta$ per imputed sample of $(X_i,Y_i)$. As a result, choosing $\kappa$ to be sufficiently large can drastically improve the convergence of predictive resampling, thereby improving computational speed through early truncation. We investigate this in Section \ref{sec:kappa} of the Appendix.  

\newpage
\section{Simulation study}\label{sec:sim}
We illustrate our methodology in a high-dimensional robust linear regression simulation with $n = 250$ and $p = 500$. The data generating process is $P^*_X = \mathcal{N}(0, I_p)$ and $p_{\beta^*}(y \mid x) = \text{Student-}t\!\left(\mu = {\beta^*}^\top x, \sigma = 1, \nu = 4\right)$, where $\beta^*$ contains 5 nonzero coefficients $(-2, -2, 2, 2, -2)$. For the parametric martingale posterior, 
we use the Student-$t$ plug-in predictive with recursive update (\ref{eq:recur_update}). We choose $D(\beta_n)$ corresponding to the Gaussian continuous spike-and-slab prior as in Table~\ref{tab:D_n_summary} with $\theta = 0.1, v_1 = 1$, and $v_0 = 0.0005$, where $\beta_n$ is initialized to the maximum a posteriori estimate. We compare to the traditional Bayesian posterior with the equivalent likelihood and prior, where partial conjugacy in this case enables the use of a Gibbs sampler.

Predictive resampling and Gibbs sampling are implemented in Julia, and run on an Apple M1 chip. We highlight that we have not utilized the inherent parallelization of predictive resampling, which would result in an even larger speed-up compared to Gibbs sampling. 
For predictive resampling, we use the same random sequence $X_{n+1:N}$ drawn from $P_X$ across predictive resampling trajectories, as precomputed terms such as $\mathcal{I}_i^{-1}X_i$ can be reused for efficiency. This is also related to the fixed-design setting, which we discuss further in Section \ref{sec:fixed_design}. Further computational details are provided in Section \ref{sec:computation} of the Appendix. 
\begin{figure}[h!]
\centering
\includegraphics[width=0.96\linewidth]{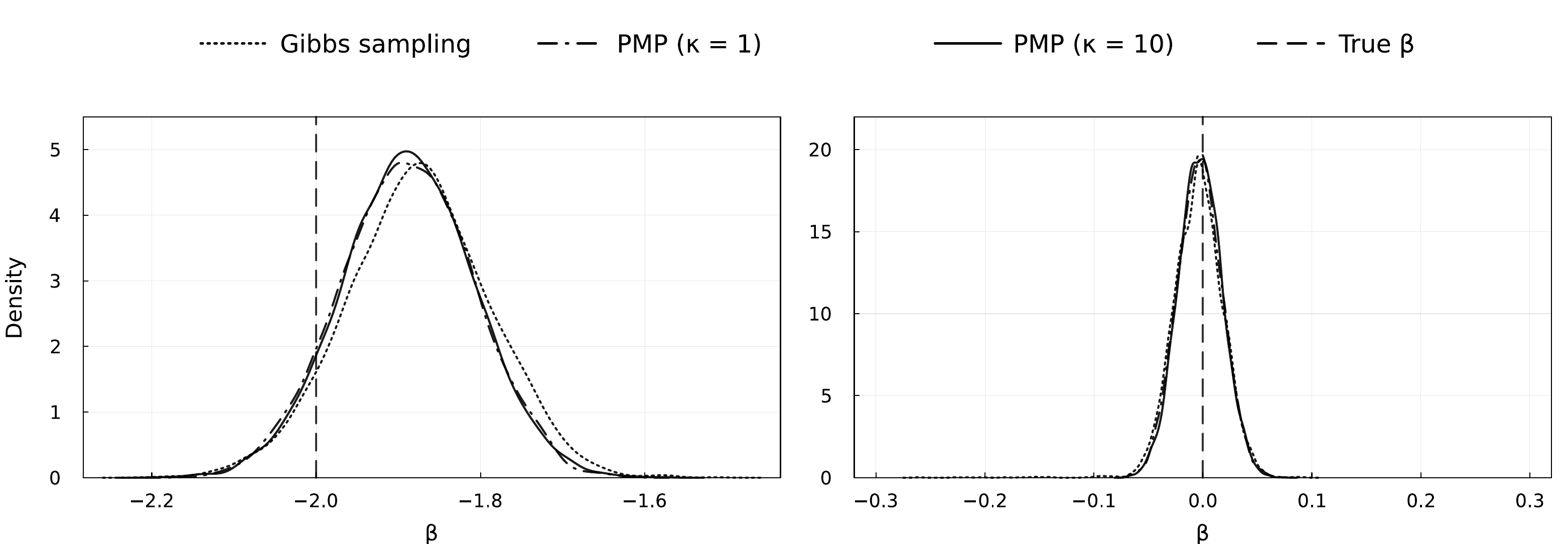}
\caption{Posteriors for selected (Left) active coefficient; (Right) inactive coefficient. PMP is the parametric martingale posterior, and $\kappa$ is the standard deviation of $P_X = \mathcal{N}(0,\kappa^2 I_p)$. } 
\label{fig:pr_vs_bayes_sim_robust}
\end{figure}

Figure~\ref{fig:pr_vs_bayes_sim_robust} illustrates that the parametric martingale posterior and Bayesian posterior are similar in terms of the posterior uncertainty for both active and inactive coefficients. We also see that weak design invariance is sufficient for robustness to the choice of $P_X$, as rescaling \(P_X\) from \(\mathcal{N}(0, \,I_p)\) to \(\mathcal{N}(0, 10^2 I_p)\) does not alter the posterior uncertainty of coefficients based on the proposed update \eqref{eq:recur_update}. As discussed earlier, this choice is accompanied with a computational benefit, as using $P_X$ with the larger covariance $\kappa = 10$ accelerates the convergence of predictive resampling, allowing early truncation at $N = p+100$ without compromising uncertainty quantification. 

Table~\ref{tab:coverage_length_time} presents coverage and interval lengths of 95\% credible intervals and run-times for 200 repeats of the simulation. The Gibbs sampler is run with $10000$ burn-in iterations followed by $25000$ samples, giving an effective sample size of approximately $5000$ for active coefficients. Compared to traditional Bayes, the parametric martingale posterior achieves similar coverage and interval lengths, while requiring significantly less computation time to produce 5000 independent posterior samples. 

\begin{table}[htbp]
\caption{Comparison of the parametric martingale posterior ($\kappa = 10$) with traditional Bayes.}
 \vspace{5mm}
\centering
\begin{small}
{\begin{tabular}{crrrr}
& \multicolumn{2}{c}{Gibbs}  &\multicolumn{2}{c}{PMP} \\[1pt]
 Set & {Cov} & {Length} & {Cov} & {Length} \\[5pt]
 Active   & 97.4  & 3.4  &   96.4  & 3.2 \\
 Inactive & 100.0 & 1.0  & 100.0 & 0.8\\[5pt]
 Max.\ SE & 0.012 & 0.004 &0.015 & 0.001  \\
 Run-time & \multicolumn{2}{c}{59.63\,s}& \multicolumn{2}{c}{0.12\,s}\\[9pt]
\end{tabular}}
\end{small}
\begin{minipage}{0.6\linewidth}
  \footnotesize 
  Cov, coverage is reported as percent; SE, standard error. 
  Interval lengths are multiplied by 10. 
  Run-time is in seconds per dataset.
\end{minipage}
\label{tab:coverage_length_time}
\end{table}

\section*{Acknowledgement}
E. Fong receives funding from the Research Grants Council of Hong Kong through the Early Career Scheme (Grant No. 27304424) and the General Research Fund (Grant No. 17306925).

\section*{Code}
Our implementation and reproducible simulation codes are available on GitHub: \url{https://github.com/wys0071101/Design_Parametric_MP}.

\bibliographystyle{apalike}
\bibliography{paper-ref}

\newpage

\begin{appendices}

\setcounter{equation}{0}
\setcounter{theorem}{0}
\setcounter{proposition}{0}
\setcounter{lemma}{0}
\setcounter{corollary}{0}
\setcounter{assumption}{0}

\renewcommand{\theequation}{\thesection.\arabic{equation}}
\renewcommand{\thetheorem}{\thesection.\arabic{theorem}}
\renewcommand{\theproposition}{\thesection.\arabic{proposition}}
\renewcommand{\thelemma}{\thesection.\arabic{lemma}}
\renewcommand{\thecorollary}{\thesection.\arabic{corollary}}
\renewcommand{\theassumption}{\thesection.\arabic{assumption}}

\renewcommand{\thetable}{\thesection.\arabic{table}}
\renewcommand{\thefigure}{\thesection.\arabic{figure}}

\section*{Summary of Appendix}
    The appendix contains the technical proofs, additional derivations and experimental results. Sections \ref{sec:proof_doob} - \ref{sec:proof_asym} cover the technical proofs of our results. Section \ref{sec:fixed_design} discusses the fixed-design case, and Section \ref{sec:ident} investigates the role of identifiability. Section \ref{sec:update_motiv} provides a general motivation for our recursive update, and Section \ref{sec:experiments} provides additional experimental details and examples.
\maketitle

\section{Proof of Proposition \ref{prop:doob_regression}}\label{sec:proof_doob}
Consider a joint probability space on $\beta$ and the observations where we first draw $\beta \sim \Pi(\cdot \mid \mathcal{F}_n)$, then generate $X_i \iid P_X$ and $(Y_i \mid X_i,\beta) \sim P_{\beta}(\cdot \mid X_i)$ for $i \geq n+1$. The joint law factorizes as
\[
P_\beta(dx,dy)
=
P_X(dx)\,P_\beta(dy\mid x).
\]
Since the marginal distribution $P_X$ is fixed and does not depend on $\beta$, the terms involving $X_{1:i}$ act as a constant factor, and Bayes' theorem thus gives the posterior for $i \geq n + 1$ as
\begin{align*}
 \Pi(d\beta \mid \mathcal{F}_i)&\propto \Pi(d\beta \mid \mathcal{F}_n)\, \prod_{j = n+1}^{i}p_{\beta}(Y_j \mid X_j) \\
 &\propto  \Pi(d\beta)\, \prod_{j = 1}^{i}p_{\beta}(Y_j \mid X_j).
\end{align*}
Marginalizing over $\beta$ recovers the predictive distribution $P_{i-1}(\cdot\mid X_i)$,
$$P(Y_i \in {\cdot} \mid \mathcal{F}_{i-1}, X_i) = \int P_{\beta}( {\cdot} \mid X_i) \, \Pi(d\beta \mid \mathcal{F}_{i-1}),$$
so this construction is equivalent to the predictive resampling scheme.

We will show Proposition \ref{prop:doob_regression} through a specialization of Doob's consistency theorem (Theorem 2.2 of \cite{miller2018detailed}) for our settings. Under predictive resampling and the assumption $E\left(\|\beta\| \mid \mathcal{F}_n\right)<\infty$, it follows from usual arguments that $(\beta_i, \mathcal{F}_i)_{i \geq n+1}$ is a uniformly integrable martingale and 
\[
\beta_i \to \beta_\infty \quad \textnormal{a.s.}
\]
To identify the limit $\beta_\infty$, we now verify the assumptions of Theorem 2.2 of \citet{miller2018detailed} for the joint model for $(X,Y)$. Let $P_\beta$ denote the joint distribution of $(X,Y)$ corresponding to the composition of the probability kernel $P_\beta(\cdot \mid x)$ and the fixed marginal $P_X$.
 
 For the identifiability condition, consider a pair $\beta\neq\beta'$. As $Y$ lies in a standard Borel space, we have $P_X$-a.e. uniqueness of the regular conditional probability kernel $P(\cdot \mid x)$ for a given joint distribution, e.g. \citet[Theorem 8.5]{kallenberg2021foundations}. 
From Assumption \ref{as:ident}, there exists a measurable set $B \subseteq \mathbb{R}^p$, which may depend on ($\beta,\beta'$), where $P_X(B)>0$ and
\[
P_\beta(\cdot\mid x)\neq P_{\beta'}(\cdot\mid x)
\quad \forall x\in B.
\]
As $P_X$ does not change with $\beta$, the above assumption must imply the joint distributions are not equal, that is $P_\beta\neq P_{\beta'}$. The joint model $\{P_\beta:\beta\in\mathbb R^p\}$ is thus identifiable.

Next, we verify measurability. By assumption, $\left\{P_\beta(\cdot \mid x):x \in \mathbb{R}^p, \beta \in \mathbb{R}^p\right\}$ is a valid probability kernel. The joint measure $P_\beta$ can be defined as the composition  of $\{P_\beta(\cdot \mid  x) : x \in \mathbb{R}^p,\beta \in \mathbb{R}^p\}$ and $P_X$, where the latter can be viewed as a constant kernel over $\beta \in \mathbb{R}^p$. \citet[Lemma 3.3]{kallenberg2021foundations} thus ensures $\{P_\beta(\cdot): \beta \in \mathbb{R}^p\}$ is a valid probability kernel, that is the map $\beta \mapsto P_\beta(C)$ for any measurable $C \subseteq  \mathbb{R}^p \times \mathcal Y$ is itself measurable.  

Hence, the joint parametric family $\{P_\beta:\beta\in\mathbb R^p\}$ satisfies the assumption required for \citet[Theorem 2.2]{miller2018detailed}, giving
$\beta_\infty = \beta$ a.s.,
and since $\beta \sim \Pi(\cdot \mid \mathcal F_n)$, we have $\beta_\infty \sim \Pi(\cdot \mid \mathcal F_n)$.

\section{Proof of Proposition \ref{prop:identifiability_sufficient}}
Let $\beta, \beta' \in \mathbb{R}^p$ with $\beta \neq \beta'$, and define
\[
v = \beta - \beta' \neq 0 .
\]
From Assumption \ref{as:suff_ident}, $E_{P_X}(XX^\top) \succ 0$, so we have that
\[
E_{P_X}\{(v^\top X)^2\} = v^\top\, E_{P_X}(XX^\top)\, v > 0 .
\]
Therefore the set
\[
B_v
=
\{x \in \mathbb R^p : v^\top x \neq 0\}
=
\{x \in \mathbb R^p :  \beta^\top x \neq  {\beta'}^\top x\},
\]
which depends on $(\beta,\beta')$, has positive probability under $P_X$, that is
\[
P_X(B_v) > 0 .
\]
For any $x \in B_v$, define
\[
t = \beta^\top x, \quad t' =  {\beta'}^\top x,
\]
where clearly $t \neq t'$. By the assumed identifiability of the one-parameter family $\{P_Y(\cdot \mid t): t \in \mathbb{R}\}$, we have
\begin{align*}
P_Y(\cdot \mid  \beta^\top x) \neq P_Y(\cdot \mid {\beta'}^\top x) \quad \forall x \in B_v.
\end{align*}
The set $B_v$ thus satisfies Assumption \ref{as:ident}.

For completion, we verify that $\{P_Y(\cdot \mid  \beta^\top x) : x \in \mathbb{R}^p,\beta \in \mathbb{R}^p\}$ is  a valid probability kernel in this case so Proposition \ref{prop:doob_regression} applies. From Section \ref{sec_setup}, $\{P_Y(\cdot \mid t) : t \in \mathbb{R}\}$ is a probability kernel by assumption. Since the map $(x,\beta) \mapsto \beta^\top x$ is measurable, the composition $\{P_Y(\cdot \mid  \beta^\top x) : x \in \mathbb{R}^p,\beta \in \mathbb{R}^p\}$ remains a valid probability kernel.

\newpage
\section{Proof of Theorem \ref{thm:predictive_invariance_mp}}\label{sec:weak_proof}
In this section, let $(\beta_i)_{i \geq n+1}$ be the sequence under predictive resampling with update rule (\ref{eq:recur_update}) and $P_i(\cdot \mid x) = P_{Y}(\cdot  \mid \beta_i^\top x)$.
To begin, we restate the preservation of the martingale property.
As $\alpha_{i-1}(X_i)$ and $\mathcal{I}_i$ only depend on $(\mathcal{F}_{i-1},X_i)$ and
\[
E\{s(\beta_{i-1}, Y_i \mid X_i)\mid\mathcal F_{i-1},X_i\}=0,
\]
 the sequence $(\beta_i,\mathcal{F}_i)_{i\geq n+1}$ is a martingale. As assumed in the main paper, $D(\beta_n)$ is chosen so that $\mathcal{I}_i$ is invertible for each $i \geq n$. We note that Assumption \ref{as:suff_ident} and the conditions on $w(t)$ are only required in the final step in Lemma \ref{lem:EIn_0}.

 First, we compute the covariance increment, where the role of the scalar correction factor $\alpha_{i-1}(X_i)$ will become clear. 
\medskip
\begin{lemma}[Covariance increment]
\label{lem_increment_cov}
Let $Z_i=\beta_i-\beta_{i-1}$ be the martingale difference for $i\geq n+1$. Then we have
\[
E\left(Z_i\,Z_i^\top\mid \mathcal F_{i-1}, X_i\right)
=
\mathcal{I}_{i-1}^{-1}-\mathcal{I}_i^{-1}.
\]
\end{lemma}
\begin{proof}
From the recursive update \eqref{eq:recur_update}, we have
\[
Z_i
=
\alpha_{i-1}(X_i)\,
\mathcal{I}_i^{-1}\,X_i\,
r\left(Y_i \mid \beta_{i-1}^\top X_i\right).
\]
Taking conditional expectations and substituting the moment condition gives
\[
E(Z_i\,Z_i^\top\mid \mathcal F_{i-1}, X_i)
=
\alpha_{i-1}(X_i)^2\, w\left(\beta_{i-1}^\top X_i\right)\,
\mathcal{I}_i^{-1}\,
X_i\,
X_i^\top\,
\mathcal{I}_i^{-1}.
\]
Recall that
\[
\mathcal{I}_i= \sum_{j=1}^i w\left(\beta_n^\top X_j\right)\,X_j X_j^\top + D(\beta_n),
\]
where $\beta_n$ and $D(\beta_n)$ are fixed throughout predictive resampling. By the Sherman-Morrison formula \citep{sherman1950adjustment}, $\mathcal{I}_i$ can be decomposed as
\[
\mathcal{I}_i^{-1}
=
\mathcal{I}_{i-1}^{-1}
-
\frac{
w\left(\beta_n^\top X_i\right)\,
\mathcal{I}_{i-1}^{-1}\,
X_i\,X_i^\top\,
\mathcal{I}_{i-1}^{-1}
}{
1+w\left(\beta_n^\top X_i\right)\,X_i^\top\, \mathcal{I}_{i-1}^{-1}\,X_i
}\cdot
\]
which implies
\[
\mathcal{I}_i^{-1}\,X_i
=
\frac{\mathcal{I}_{i-1}^{-1}\,X_i}
{1+w\left(\beta_n^\top X_i\right)\,X_i^\top\, \mathcal{I}_{i-1}^{-1}\,X_i}\cdot
\]
Substituting the above expression into the covariance formula and using the definition of $\alpha_{i-1}(X_i)$, the factor $w(\beta_{i-1}^\top X_i)$ cancels out and yields
\[
E\left(Z_i \,Z_i^\top \mid \mathcal F_{i-1}, X_i\right)
=
\frac{
w\left(\beta_n^\top X_i\right)\,\mathcal{I}_{i-1}^{-1}\,
X_i\,X_i^\top\,
\mathcal{I}_{i-1}^{-1}
}{
1+w\left(\beta_n^\top X_i\right)\,X_i^\top\, \mathcal{I}_{i-1}^{-1}\,X_i
}\cdot
\]
Recognizing the right-hand side as the update term in the Sherman-Morrison formula for $\mathcal{I}_i^{-1}$, we obtain the critical simplification: 
\[
E\left(Z_i Z_i^\top \mid \mathcal F_{i-1}, X_i\right)
=
\mathcal{I}_{i-1}^{-1} - \mathcal{I}_i^{-1}
\]
as desired.
\end{proof}
Next, the covariance increment will allow us to calculate the posterior covariance using a telescoping argument. 
\begin{lemma}[Covariance]\label{lem:post_cov}
    The posterior covariance takes the form
    \begin{align}\label{eq:cov_N}
       \textnormal{Cov}(\beta_N \mid \mathcal{F}_n)
&= \mathcal{I}_n^{-1} - E\left(\mathcal{I}_N^{-1} \mid \mathcal{F}_n\right).
    \end{align}
\end{lemma}
\begin{proof}
    We begin with the well-known orthogonal decomposition for martingales:
    \begin{align*}
               \textnormal{Cov}(\beta_N \mid \mathcal{F}_n)  = \sum_{i = n+1}^N E\left(Z_i Z_i^\top \mid \mathcal{F}_n\right).
    \end{align*}
    Next,  express each term in the sum with the tower property as
    \begin{align*}
        E\left(Z_i Z_i^\top \mid \mathcal{F}_n\right) = E\left\{E\left(Z_i Z_i^\top \mid \mathcal{F}_{i-1}, X_i\right) \mid  \mathcal{F}_n \right\}
    \end{align*}
    which then gives
    \begin{align*}
        \textnormal{Cov}(\beta_N \mid \mathcal{F}_n)  = E\left\{\sum_{i = n+1}^N E\left(Z_i Z_i^\top \mid \mathcal{F}_{i-1}, X_i\right)  \mid \mathcal{F}_n\right\}.
    \end{align*}
    From Lemma \ref{lem_increment_cov}, a telescoping argument gives
    \begin{align*}
        \sum_{i = n+1}^N E\left(Z_i Z_i^\top \mid \mathcal{F}_{i-1}, X_i\right) = \mathcal{I}_n^{-1} - \mathcal{I}_N^{-1}.
    \end{align*}
    Finally, since $\mathcal{I}_n^{-1}$ is fixed (as it only depends on the fixed $X_{1:n}$ and $\beta_n$), we have the desired result. Note that since the covariance is non-negative, we must have $ E\left(\mathcal{I}_N^{-1} \mid \mathcal{F}_n\right) \preceq \mathcal{I}_n^{-1}$.
\end{proof}

We can already show the martingale is bounded in $L^2$.
\begin{lemma}[Bounded in $L^2$]
\label{lem_L2}
The martingale $(\beta_i,\mathcal{F}_i)_{i\geq n+1}$ is bounded in $L^2$ and 
thus converges a.s. and in $L^2$ to a limit $\beta_\infty$.
\end{lemma}
\begin{proof}
From Lemma \ref{lem:post_cov}, we have
\begin{align*}
    E\!\left(
\|\beta_N-\beta_n\|^2
\mid
\mathcal F_n\right) &= \text{tr}\left( \mathcal{I}_n^{-1}\right) - E\left\{\text{tr}\left( \mathcal{I}_N^{-1}\right)\mid \mathcal{F}_n \right\}\\
&\leq \text{tr}\left( \mathcal{I}_n^{-1}\right)
\end{align*}
since $\mathcal{I}_N^{-1}$ is positive semidefinite with non-negative trace. 
It then follows that
\[
\sup_{N\geq n+1}
E\!\left(
\|\beta_N\|^2
\mid
\mathcal F_n
\right)
\le
\|\beta_n\|^2
+
\,\operatorname{tr}\!\left(\mathcal{I}_n^{-1}\right)
<\infty
\]
where we have used the orthogonality of martingale increments in $L^2$.
\end{proof}

Finally, we handle the term $E(\mathcal{I}_N^{-1} \mid \mathcal{F}_n)$. 
\begin{lemma}\label{lem:EIn_0}
Suppose Assumption \ref{as:suff_ident} holds, and $w(t)$ is continuous and strictly positive for all $t \in \mathbb{R}$. We then have
\begin{align*}
    \lim_{N \to \infty} E(\mathcal{I}_N^{-1} \mid \mathcal{F}_n)  = 0. 
\end{align*}
\end{lemma}
\begin{proof}
    In the following, we consider the asymptotic behavior of $\mathcal{I}_N$, and highlight again that the initial dataset $(X_{1:n}, Y_{1:n})$ and the resulting estimator $\beta_n$ are fixed. From Assumption \ref{as:suff_ident}, $E_{P_X}\left(XX^\top\right) \succ 0$, so there must exist a sufficiently large closed ball $B_r = \{x \in \mathbb{R}^p: \|x\| \leq r\}$ such that
    \begin{align*}
        E_{P_X}\left\{XX^\top \mathbbm{1}\left(X \in B_r\right)\right\} \succ 0.
    \end{align*}
    To see the existence of such an $r$, consider taking  $r\to \infty$, where dominated convergence gives $E_{P_X}\left\{XX^\top \mathbbm{1}\left(X \in B_r\right)\right\}  \to E_{P_X}\left(XX^\top\right)$. Continuity of eigenvalues then gives the existence.

    Choose $r$ sufficiently large as discussed, and consider the matrix
    \begin{align*}
        \mathcal{I}_{B_r} = E_{P_X}\left\{w(\beta_n^\top X)\, XX^\top \mathbbm{1}\left(X \in B_r\right)\right\}.
    \end{align*}
    From the continuity of $w(t)$ and $0 < w(t) < \infty$, there exist $c_l > 0$ and $c_u < \infty$ such that
    \begin{align*}
    c_l \leq w(\beta_n^\top x) \leq c_u
    \end{align*}
    for all $x \in B_r$ since $\beta_n^\top x$ lies in a compact set. We then immediately have that
    \begin{align*}
        c_l E_{P_X}\left\{XX^\top \mathbbm{1}\left(X \in B_r\right)\right\} \preceq\mathcal{I}_{B_r} \preceq c_u E_{P_X}\left\{XX^\top\right\},
    \end{align*}
    so $\mathcal{I}_{B_r}$ is finite and positive definite. 

    Now let
\[
A_i=w\left(\beta_n^\top X_i\right)\,X_i\,X_i^\top \mathbbm{1}\left(X_i \in B_r\right)
\]
which are independent and identically distributed random matrices with mean $E_{P_X}(A_i) = \mathcal{I}_{B_r}$. Next, 
\[
\mathcal{I}_N
\succeq
\mathcal{I}_n+\sum_{i=n+1}^N A_i \quad \text{a.s.}
\]
since $w\left(\beta_n^\top X_i\right)\,X_i\,X_i^\top \mathbbm{1}\left(X_i \in B_r^c\right)$ is positive semidefinite. Hence
\[
\frac{1}{N}\,\mathcal{I}_N \succeq \frac{\mathcal{I}_n}{N} + \frac{N-n}{N}
\left(\frac{1}{N-n}\sum_{i=n+1}^N A_i \right)\quad \text{a.s.}
\]
Next, applying the strong law of large numbers entry-wise gives
\[
\frac{1}{N-n}\sum_{i=n+1}^N A_i \to \mathcal{I}_{B_r} \quad \text{a.s.}
\]
 Since $N^{-1}{\mathcal{I}_n}\to0$ and ${(N-n)}/{N}\to1$ as $N\to\infty$, and $\mathcal{I}_{B_r} \succ 0$, Weyl's inequality gives 
\[
\liminf_{N\to \infty} \lambda_{\text{min}}\left(N^{-1}\,\mathcal{I}_N\right) >0  \quad \text{a.s.}
\]
Since $\mathcal{I}_N$ is invertible for each $N$, the above implies 
$$\|\mathcal{I}_N^{-1} \|_2= N^{-1}\,  \|(N^{-1}\,\mathcal{I}_N)^{-1}\|_2 \to 0\quad \text{a.s.}$$ 
where $\|\cdot \|_2$ is the spectral norm. Note that for all $N \geq n+1$, we have $\mathcal{I}_N \succeq \mathcal{I}_n$ a.s. which implies
$$\|\mathcal{I}_N^{-1}\|_2 \le \|\mathcal{I}_n^{-1}\|_2 \quad \text{a.s.}$$
 Since $\mathcal{I}_n$ is non-singular, $\|\mathcal{I}_n^{-1}\|_2$ is a valid integrable dominator. Therefore, by dominated convergence, we have
$$\lim_{N \to \infty} E\!\left( \mathcal{I}_N^{-1} \mid \mathcal{F}_n \right) = E\!\left( \lim_{N \to \infty} \mathcal{I}_N^{-1} \mid \mathcal{F}_n \right) = 0$$
as desired.
\end{proof}

Theorem \ref{thm:predictive_invariance_mp} can now be shown. As the martingale is uniformly integrable, we immediately have
\begin{align*}
    E(\beta_\infty \mid \mathcal{F}_n) = \beta_n,
\end{align*}
which is the posterior mean part of the result. Since $\beta_i \to \beta_\infty$ in $L^2$, the covariance of $\beta_i$ converges to the covariance of $\beta_\infty$.
Taking the limit of (\ref{eq:cov_N}) gives
\[
\operatorname{Cov}(\beta_\infty\mid\mathcal F_n)
=
\lim_{N\to\infty}
\operatorname{Cov}(\beta_N\mid\mathcal F_n)
=
\mathcal{I}_n^{-1}.
\]
This limit is then independent of the predictive design distribution $P_X$.

\section{Proof of Theorem \ref{thm:as_invar}}\label{sec:proof_asym}
\subsection{General lemmas}
We begin with some useful lemmas for finite martingale difference arrays. Let $(Z_{ni}, \mathcal{F}_{ni})_{1\le i\le N_n, n \geq 1}$ be a general square-integrable finite martingale difference array, where $Z_{ni}$ is a scalar. As usual, we have
\begin{align}\label{eq:var_def}
    s_{n}^2 = \sum_{i = 1}^{N_n} E\left(Z_{ni}^2\right), \quad V_{n}^2 = \sum_{i = 1}^{N_{n}} E\left(Z_{ni}^2 \mid \mathcal{F}_{n,i-1}\right).
\end{align}
The first is an obvious connection between the conditional Lyapunov and Lindeberg conditions (see \cite{alj2014conditions}).  
\begin{lemma}[Conditional Lyapunov condition]\label{lem:condit_lyapunov}
Suppose that the following condition holds true for some $\delta > 0$:
\begin{align*}
    s_{n}^{-(2 + \delta)}\, \sum_{i = 1}^{N_n} E\left(|Z_{ni}|^{2 + \delta} \mid \mathcal{F}_{n,i-1} \right)\to 0 \quad \textnormal{in probability.}
\end{align*}
 Then the conditional Lindeberg condition holds true, that is
\begin{align*}
    s_{n}^{-2}\, \sum_{i = 1}^{N_{n}} E\left\{Z_{ni}^2 \mathbbm{1}\left(Z_{ni}^2 > \varepsilon^2\, s_{n}^2 \right) \mid \mathcal{F}_{n,i-1}\right\} \to 0 \quad \textnormal{in probability}
\end{align*}
 for all $\varepsilon > 0$.
\end{lemma}
\begin{proof}
Following the usual argument, $Z_{ni}^2 > \varepsilon^2\, s_{n}^2$ implies 
\begin{align*}
    \left|\frac{Z_{ni}}{\varepsilon\, s_{n}}\right|^{\delta} > 1.
\end{align*}
We thus have 
\begin{align*}
    \mathbbm{1}\left(Z_{ni}^2 > \varepsilon^2\, s_{n}^2 \right)  \leq\left|\frac{Z_{ni}}{\varepsilon s_{n}}\right|^{\delta}.
\end{align*}
We thus have
\begin{align*}
    {s}_{n}^{-2}\,\sum_{i = 1}^{N_{n}}\, E\left\{Z_{ni}^{2}\, \mathbbm{1}\left(Z_{ni}^2 > \varepsilon^2 \,s_{n}^2 \right) \mid \mathcal{F}_{n,i-1}\right\} \leq \frac{s_{n}^{-(2+\delta)}}{\varepsilon^{\delta}}\,\sum_{i = 1}^{N_{n}}\,E\left(|Z_{ni}|^{2 + \delta} \mid \mathcal{F}_{n,i - 1}\right) \quad \text{a.s.}
\end{align*}
giving the result.
\end{proof}

We now consider a sufficient condition for the Lyapunov condition with $\delta = 2$.
\begin{lemma}\label{lem:condit_lyapunov_suff}
    Let $\widetilde{Z}_{ni} = (n+i)\,Z_{ni}$, and suppose $s_{n}^{-2} = O(n)$. Let
    \begin{align*}
        W_{n} = \sup_{1\le i\le N_{n}} E\left(\widetilde{Z}_{ni}^{4} \mid \mathcal{F}_{n,i-1} \right).
    \end{align*}
If the sequence $(W_{n})_{n \geq 1}$ is bounded in probability, then the conditional Lyapunov condition holds for $\delta = 2$.
\end{lemma}
\begin{proof}
 For $\delta = 2$, the conditional Lyapunov term is proportional to
    \begin{align*}
      L_{n} &= n^2 \,\sum_{i = 1}^{N_{n}} (n+i)^{-4} \,E\left(\widetilde{Z}_{ni}^4 \mid \mathcal{F}_{n,i-1}\right).
    \end{align*}
    We clearly have
    \begin{align*}
    L_n\leq W_n\,n^2\,\sum_{i=1}^{\infty}(n+i)^{-4}\quad \text{a.s.}
    \end{align*}
    Since $n^2\,\sum_{i=1}^{\infty}(n+i)^{-4}=O(n^{-1}),$ and $W_n=O_P(1)$, we have $L_{n} \to 0$ in probability.
\end{proof}

For completion, we also state a specialized form of the martingale central limit theorem used here, which is based on \citet[Theorem 6.1]{hausler2015stable}.
\begin{theorem}\label{th:mart_CLT}
   Let $(Z_{ni}, \mathcal{F}_{ni})_{1\le i\le N_n,\, n\ge1}$ be a square-integrable finite martingale difference array, and let $s_{n}^2$ and $V_{n}^2$ be defined according to (\ref{eq:var_def}). Assume that 
    \begin{align*}
    s_{n}^{-4}\, \sum_{i = 1}^{N_n} E\left(Z_{ni}^{4} \mid \mathcal{F}_{n,i-1} \right)\to 0 \quad \textnormal{in probability},
\end{align*}
and 
\begin{align*}
    s_{n}^{-2}V_{n}^2 \to 1 \quad \textnormal{in probability}.
\end{align*}
Under the above, we have
\begin{align*}
    s_{n}^{-1}\sum_{i = 1}^{N_n} Z_{ni} \to \mathcal{N}(0,1)\quad \textnormal{in distribution}.
\end{align*}
\end{theorem}
\begin{proof}
We apply \citet[Theorem 6.1]{hausler2015stable} with the martingale difference as $s_{n}^{-1}Z_{ni}$, then use Lemma \ref{lem:condit_lyapunov} to replace the conditional Lindeberg condition with the conditional Lyapunov condition. Finally, as $\eta^2 = 1$ is a constant, its measurability is automatic; see \citet[Theorem 6.1(b)]{hausler2015stable}.
\end{proof}

\subsection{Setup and notation}
For the proofs in this section, we will tweak our notation as follows compared to the rest of the Appendix. For each $n$ and $i\ge1$, we consider
\begin{align*}
    &{X}_{i} \sim P_X, \, \quad ({Y}_{ni} \mid {X}_{i}) \sim p_{\beta_{n,i-1}}(\cdot \mid {X}_{i}),\\
    \beta_{ni} &= \beta_{n,i-1} + \alpha_{n,i-1}({X}_{i})\, \mathcal{I}_{ni}^{-1}\,X_i \,  r\left(Y_{ni} \mid \beta_{n,i-1}^\top X_i\right)
\end{align*}
where
\begin{align}\label{eq:I_ni}
    \mathcal{I}_{ni} &= \mathcal{I}_n + \Sigma_{ni}\\
    \mathcal{I}_n = \sum_{j = 1}^n w\left(\beta_n^\top X^*_j\right)\, X_j^* {X_j^*}^\top &+ D(\beta_n), \quad 
    \Sigma_{ni} = \sum_{j = 1}^i w\left(\beta_n^
    \top X_j\right)\, X_j X_j^\top,\nonumber
\end{align}
and
\begin{align*}
     \alpha_{ni}(x) &=
\left[
\frac{
w\left(\beta_n^\top\, x\right)
}{
w\left( \beta_{ni}^\top\, x\right)
}\,\left\{1 + w\left(\beta_n^\top\, x\right)\, x^\top\, \mathcal{I}_{ni}^{-1}\, x\right\}
\right]^{1/2}\cdot
\end{align*}
The index $i$ is now indicating the imputed observations \textit{after} index $n$, which is reflected in the notation $Y_{ni}$, $\beta_{ni}$, $\mathcal{I}_{ni}$ and $\alpha_{ni}(x)$.   Note that for each \(n\), we will set \(\beta_{n0}=\beta_n\), \(\mathcal I_{n0}=\mathcal I_n\) as the initial estimates. We will also write the real observations as $(X_{1:n}^*, Y_{1:n}^*)$ to differentiate it from the imputed observations. Henceforth, we define the filtration for our array as
\begin{align*}
    \mathcal{F}_{ni} = \sigma\left\{(X_1,Y_{n1}),\ldots, (X_i, Y_{ni})\right\}
\end{align*}
where $(X_{1:n}^*, Y_{1:n}^*)$ are treated as deterministic. Notice that we use the same sequence ${X}_i$ for each row in the martingale array for convenience in the proof. 

Let $(\beta_{ni}, \mathcal{F}_{ni})_{i \geq 1, n \geq 1}$ denote the infinite martingale array.
Since Theorem \ref{thm:predictive_invariance_mp} holds for each $n\geq 1$ under our assumptions,  each row $(\beta_{ni}, \mathcal{F}_{ni})_{i \geq 1}$ is a martingale bounded in $L^2$ and has an almost sure limit $\beta_{n\infty}$.  In order to apply the martingale central limit theorems, we will project the vectors to scalars and utilize the Cram\'{e}r-Wold theorem. For a fixed vector $u \in \mathbb{R}^p$, define the scalar
\begin{align}\label{eq:Zni}
    Z_{ni} =u^\top  \left(\beta_{ni} - \beta_{n,i-1}\right)
\end{align}
where we omit the dependence on $u$ for notational simplicity. We will return to the multivariate case in Section \ref{sec:trunc}. 

Although we are interested in the infinite limit \( \beta_{n\infty} = \lim_{i \to \infty}\beta_{ni}\), unlike in \cite{fong2026asymptotics}, it will be convenient to first truncate each row $n$ at a finite index \(N_n > n\), where we will assume  
$$N_n/n \to \infty.$$
We will then show that the difference between the finite and infinite martingale arrays is negligible in Section \ref{sec:trunc}. For the remainder of this proof, we will study the specific martingale difference array $(Z_{ni}, \mathcal{F}_{ni})_{1\leq i \leq N_n, n \geq 1}$ as defined above.

\subsection{Proof outline}
We will verify the assumptions required for Theorem \ref{th:mart_CLT}. Unlike \citet{fong2026asymptotics}, who utilize an {unconditional} Lindeberg/Lyapunov condition, here we will require the {conditional} Lindeberg/Lyapunov condition so we can use the strong law of large numbers to handle the term $\mathcal{I}_{ni}$. The verification of the variance condition also requires a delicate application of \citet[Proposition 6.16]{hausler2015stable} to bridge an enlarged filtration to the standard filtration, in order to leverage the telescoping form of the variance as given by Lemma \ref{lem_increment_cov}. 

As in \citet{fong2026asymptotics}, we will assume a deterministic condition first, before generalizing to the a.s. case in Section \ref{sec:trunc}. Below, we have the deterministic version of Assumption \ref{as:In_conv}.
\begin{assumption}\label{as:In_conv_det}
Let the following hold deterministically: $\mathcal{I}_n \succ 0$ for all $n$, $\beta_n \to \beta^*$, and $n^{-1}\,\mathcal{I}_n \to \mathcal{I}_{P_X^*} \succ 0$.
\end{assumption}
We highlight here that the limits $\beta^*$ and $\mathcal{I}_{P_X^*}$ are arbitrary and need not necessarily correspond to the true values. We explicitly define them here as this will be useful to refer to later on. 

\subsection{Intermediate lemmas}
In the proofs to come, we will require control over $w(t)$ and $\mathcal{I}_{ni}^{-1}$. We give two useful lemmas here.
\begin{lemma}\label{lem:compact_wt}
    Under Assumptions \ref{as:In_conv_det} and \ref{as:lyapunov_hard}, there exist constants $c_l > 0$ and $c_u <\infty$ such that
$$c_l\leq w(\beta_n^\top\, x)\leq c_u$$
 for all $n \geq 1$ and $x \in \textnormal{supp}\left(P_X\right)$. 
\end{lemma}
\begin{proof}
    To begin, we note that $\beta_n \to \beta^*$ from Assumption \ref{as:In_conv_det} and $\text{supp}\left(P_X\right)$ is compact from Assumption \ref{as:lyapunov_hard},  so $\beta_n^\top\, x$ lies in a compact interval of $\mathbb{R}$ for all $n\geq 1$ and $x \in \textnormal{supp}\left(P_X\right)$. As $w(t)$ is continuous and satisfies $0 < w(t) < \infty$ for all $t \in \mathbb{R}$, $w(\beta_n^\top x)$ lies in a compact interval of $\mathbb{R}$ for all $n \geq 1$ and $x \in \text{supp}\left(P_X\right)$.
\end{proof}

\begin{lemma}\label{lem:sup_Ini}
Consider the non-negative random variable                          
 \begin{align}\label{eq:Rni}
 R_{ni} =\left\|(n+ i + 1)\, \mathcal{I}_{ni}^{-1} \right\|_2
 \end{align}
 which is finite a.s. for each $(n \geq 1,i \geq 0)$. Under Assumptions \ref{as:suff_ident}, \ref{as:In_conv_det} and either \ref{as:lyapunov_easy} or \ref{as:lyapunov_hard}, we have
 \begin{align*}
 \sup_{n \geq 1, i \geq 0} R_{ni }< \infty \quad \textnormal{a.s.}
 \end{align*}
\end{lemma}
\begin{proof}
   First, finiteness of $R_{ni}$ follows from $\mathcal{I}_{ni} \succeq \mathcal{I}_n$ a.s. and the fact that $\mathcal{I}_n \succ 0$ for each $n$ from Assumption \ref{as:In_conv_det}. We now tackle the proof under the two sets of assumptions separately.
    
    \textit{Constant $w(t)= w_0$ case}:  Since $\|\mathcal{I}_{ni}^{-1}\|_2 = 1/\lambda_{\min}\left(\mathcal{I}_{ni}\right)$, we have
$$ R_{ni} = \frac{n+i+1}{\lambda_{\min}(\mathcal{I}_{ni})}\cdot $$
Next, we write for $n \geq 1$ and $i \geq 1$:
$$ a_n = \lambda_{\min}(n^{-1}\,\mathcal{I}_n), \quad b_i = \lambda_{\min}(i^{-1}\,\Sigma_i), \quad \Sigma_i = w_0\,\sum_{j=1}^i X_j \,X_j^\top$$
and we define $b_0 = 0$. From Assumption \ref{as:In_conv_det}, $a_n > 0$ for each $n$ and $a_n \to a > 0$, so there exists $\varepsilon > 0$ such that 
$$a_n \geq \varepsilon \quad \text{for all } n \geq 1.$$
Next, Weyl's inequality gives $\lambda_{\min}(\mathcal{I}_{ni}) \geq na_n + i b_i$
which gives the following for $n \geq 1$ and $i \geq 0$:
\begin{align*}
    R_{ni} &\leq \frac{n+ 1}{n \varepsilon + i b_i} + \frac{i}{n \varepsilon + i b_i}\\
    & \leq \frac{2}{\varepsilon} + \frac{i}{\varepsilon + i b_i}\cdot
\end{align*}
Note that the above does not depend on $n$, so it is also a valid upper bound for $\sup_{n \geq 1} R_{ni}$. From Assumption \ref{as:suff_ident} and the strong law of large numbers, we have $b_i \to b > 0$ a.s., so 
\begin{align*}
    \lim_{i \to \infty}\frac{i}{\varepsilon + i b_i} = \frac{1}{b} \quad \text{a.s.}
\end{align*}
As $i/(\varepsilon + i b_i) < \infty$ for each $i\geq 0$ and it converges to a bounded limit a.s., we have
\begin{align*}
    \sup_{n \geq 1, i \geq 0} R_{ni} < \infty \quad \text{a.s.}
\end{align*}
\textit{General $w(t)$ case}:
The general case under Assumption \ref{as:lyapunov_hard} is a minor extension of the above. In this case, we note that Lemma \ref{lem:compact_wt} gives the existence of $c_l > 0$ such that
 \begin{align*}
     \Sigma_{ni}  = \sum_{j = 1}^i w(\beta_n^\top\,X_j) \, X_j\, X_j^\top \succeq c_l\, \sum_{j = 1}^i X_j X_j^\top \quad \text{a.s.}
 \end{align*}
As a result, we have 
\begin{align*}
    \mathcal{I}_{ni} \succeq \mathcal{I}_n + c_l \sum_{j = 1}^i X_j X_j^\top \quad \text{a.s.}
\end{align*}
Therefore, under Assumption \ref{as:In_conv_det} and Assumption \ref{as:suff_ident}, the proof is identical to the constant $w(t) = w_0$ case but with $c_l$ instead of $w_0$. 
\end{proof}

\subsection{Conditional Lyapunov condition}
The main result for the conditional Lyapunov condition is as follows.
\begin{lemma}\label{lem:pmp_lyapunov}
    Under Assumptions \ref{as:suff_ident}, \ref{as:In_conv_det} and either \ref{as:lyapunov_easy} or \ref{as:lyapunov_hard}, the conditional Lyapunov condition with $\delta = 2$ holds true.
\end{lemma}
\begin{proof}
In our setting, we have 
\begin{align}\label{eq:tilde_Zni}
\widetilde Z_{ni} :=  (n+i)\,Z_{ni} = (n+i)  \,\frac{\sqrt{w\left(\beta_n^\top\, X_i\right)}\,
u^\top \mathcal I_{n,i-1}^{-1}\,X_i
}{
\sqrt{1+w\left(\beta_n^\top\, X_i\right)\,X_i^\top\,\mathcal I_{n,i-1}^{-1}\,X_i
}}\, \frac{r(Y_{ni}\mid \beta_{n,i-1}^\top X_i)}{\sqrt{w\left(\beta_{n,i-1}^\top\, X_i\right)}}
\end{align}
where $Z_{ni}$ is defined in (\ref{eq:Zni}).
 It is not too hard to show that
\begin{align}\label{eq:gen_EZ}
    E\left( \widetilde{Z}_{ni}^4  \mid \mathcal{F}_{n,i-1}, X_i \right) \leq \|u\|^4 R_{n,i-1}^4 \,\, w\left(\beta_n^\top\, {X}_{i}\right)^2 \|{X}_{i}\|^4\,  k\left(\beta_{n,i-1}^\top X_i\right) \quad \text{a.s.}
\end{align}
where $R_{ni}$ is defined in (\ref{eq:Rni}). We now seek to apply Lemma \ref{lem:condit_lyapunov_suff}, where we note that $s_{n}^{-2} = O(n)$ from Lemma \ref{lem:s_nNn} to come. We now consider the two sets of assumptions separately to show 
$$\sup_{1 \leq i \leq N_n}E\left(\widetilde{Z}_{ni}^4 \mid \mathcal{F}_{n,i-1}\right) = O_p(1).$$

\textit{Constant $w(t) = w_0$ and $k(t) = k_0$ case}:
We begin with verifying the conditional Lyapunov condition under Assumption \ref{as:lyapunov_easy}.  In this case, taking the expectation of (\ref{eq:gen_EZ}) over $X_i$ gives 
\begin{align*}
E\left(\widetilde{Z}_{ni}^4 \mid \mathcal{F}_{n,i-1} \right) \leq  K_1\,\|u\|^4\, R_{n,i-1}^4 \quad \text{a.s.}
\end{align*} where $K_1 = w_0^2\,k_0\,E_{P_X}(\|X\|^4) < \infty$ is a constant independent of $n$ and $i$ following Assumption \ref{as:lyapunov_easy}. From Lemma \ref{lem:sup_Ini}, taking the supremum over $(n,i)$ gives
\begin{align*}
    \sup_{n \geq 1}\left\{\sup_{1 \leq i \leq N_n}E\left(\widetilde{Z}_{ni}^4 \mid \mathcal{F}_{n,i-1} \right)\right\} < \infty \quad \text{a.s.}
\end{align*}
 so
 \begin{align*}
     \sup_{1 \leq i \leq N_n}E\left(\widetilde{Z}_{ni}^4 \mid \mathcal{F}_{n,i-1} \right) = O_P(1)
 \end{align*}
 and by Lemma \ref{lem:condit_lyapunov_suff}, the conditional Lyapunov condition holds.

\textit{General $w(t)$ and $k(t)$ case}:
We now consider the general case under Assumption \ref{as:lyapunov_hard}. Let $\mathcal{X} = \text{supp}\left(P_X\right)$ be the compact support. Taking the expectation of (\ref{eq:gen_EZ}) over $X_i$ gives
 \begin{align*}
     E\left( \widetilde{Z}_{ni}^4  \mid \mathcal{F}_{n,i-1} \right) \leq K_2\,\|u\|^4 \, R_{n,i-1}^4\,\tilde{k}(\beta_{n,i-1}) \quad \text{a.s.}
 \end{align*}
 where $K_2 = c_u^2\,\sup_{x\in\mathcal X}\|x\|^4$ following Lemma \ref{lem:compact_wt} and 
 \begin{align*}
     \tilde{k}(\beta):=  \sup_{x \in \mathcal{X}} k\left(\beta^\top\, x\right).
 \end{align*}
We can handle the Fisher information and the kurtosis term separately as
\begin{align*}
   \sup_{i \geq 1} E\left( \widetilde{Z}_{ni}^4  \mid \mathcal{F}_{n,i-1} \right) \leq K_2\,\|u\|^4  \, \left(\sup_{i \geq 0} R_{ni}^4 \right)\, \left(\sup_{i \geq 0}\tilde{k}(\beta_{ni})\right) \quad \text{a.s.}
\end{align*}
For the Fisher information term, Lemma \ref{lem:sup_Ini} again gives $\sup_{n \geq 1, i \geq 0} R_{ni}^4< \infty$ a.s. so $\sup_{i \geq 0} R_{ni}^4 = O_P(1)$.

For the kurtosis term, we note that since $k(\beta^\top\, x)$ is continuous in $(\beta,x)$ and $\mathcal{X}$ is compact and does not depend on $\beta$, Berge's maximum theorem gives that $\tilde{k}(\beta)$ is also continuous. Next, consider the envelope function
\begin{align*}
    U(s) = \sup_{\| \beta \|^2 \leq s} \tilde{k}(\beta)
\end{align*}
which is finite for each $s \geq 0$ from the continuity of $\tilde{k}$ and compactness of the closed ball.  Crucially, the definition of $U(s)$ implies
\begin{align*}
    \sup_{i \geq 0} \tilde{k}(\beta_{ni}) \leq U\left(\sup_{i \geq 0} \| \beta_{ni}\|^2\right).
\end{align*}
Combining the above with the fact that $U(s)$ is non-decreasing then gives
\begin{align*}
    {P}\left(\sup_{i \ge 0} \tilde{k}(\beta_{ni}) > U(M)\right) &\leq  {P}\left\{U\left(\sup_{i \geq 0} \| \beta_{ni}\|^2\right) > U(M)\right\}\\
    &\le {P}\left(\sup_{i \ge 0} \|\beta_{ni}\|^2 > M\right).
\end{align*}
From the above, to show $\sup_{i \geq 0}\tilde{k}(\beta_{ni}) = O_P(1)$, it is sufficient to show that $\sup_{i \geq 0} \| \beta_{ni}\|^2 = O_P(1)$, which we now do using Doob's $L_p$ inequality.

Since $(\beta_{ni},\mathcal{F}_{ni})_{i \geq 1}$ is a (vector) martingale for each $n$, convexity of the norm and Jensen's inequality imply $(\|\beta_{ni}\|,\mathcal{F}_{ni})_{i \geq 1}$ is a positive (scalar) sub-martingale for each $n$. We can thus apply Doob's $L_p$ inequality, e.g. \citet[Theorem 4.4]{kuhn2023maximal}, to get
\begin{align*}
    E\left(\sup_{i \geq 1}\|\beta_{ni}\|^2 \right) = E\left\{\left(\sup_{i \geq 1}\|\beta_{ni}\|\right)^2\right\}\leq 4 \sup_{i \geq 1} E\left(\| \beta_{ni} \|^2 \right)
\end{align*}
where the right-most term is finite as $(\beta_{ni},\mathcal{F}_{ni})_{i \geq 1}$ is a martingale bounded in $L^2$. The proof of Lemma \ref{lem_L2} further gives
\begin{align*}
    \sup_{i \geq 1} E\left(\| \beta_{ni} \|^2 \right) \leq \|\beta_n\|^2 + \text{tr}\left( \mathcal{I}_n^{-1}\right).
\end{align*}
From Assumption \ref{as:In_conv_det},  $n\,\mathcal{I}_n^{-1} \to \mathcal{I}_{P_X^*}^{-1}$ and $\beta_n \to \beta^*$ implies
\begin{align*}
    \|\beta_n\|^2 + \text{tr}\left( \mathcal{I}_n^{-1}\right) \to \|\beta^*\|^2,
\end{align*}
so 
$$\sup_{n \geq 1}  E\left(\sup_{i \geq 1}\|\beta_{ni}\|^2\right) < \infty$$ and Markov's inequality implies 
$$\sup_{i \geq 1}\|\beta_{ni}\|^2 = O_P(1).$$ 
As $\beta_{n0} = \beta_{n}$ is deterministic, putting all the above together gives
 \begin{align*}
\sup_{i \geq 0}\tilde{k}\left(\beta_{ni}\right) = O_P(1),
 \end{align*}
 which in turn implies 
\[
\sup_{1\le i\le N_n}
E\left(
\widetilde Z_{ni}^4
\mid
\mathcal F_{n,i-1}
\right)
=
O_P(1).
\]
By Lemma \ref{lem:condit_lyapunov_suff}, the conditional Lyapunov condition holds.
\end{proof}
\subsection{Variance condition}
Before verifying the variance condition, a useful intermediate result is the behaviour of $s_{n}^2$ as $n\to \infty$.
\begin{lemma}\label{lem:s_nNn}
     Under Assumptions \ref{as:suff_ident}, \ref{as:In_conv_det} and either \ref{as:lyapunov_easy} or \ref{as:lyapunov_hard}, we have
     \begin{align*}
        n\, s_{n}^2 \to u^\top\, \mathcal{I}_{P_X^*}^{-1}\, u.
     \end{align*}
\end{lemma}
\begin{proof}
    Recall that
\[
s_{n}^2
=
\sum_{i=1}^{N_n}
E\left(Z_{ni}^2\right).
\]
By Lemma~\ref{lem_increment_cov}, we have
\begin{align}\label{eq:telescop_var}
\sum_{i=1}^{N_n}
E\left(Z_{ni}^2\mid \mathcal F_{n,i-1},X_i\right)
=
u^\top\,
\left(
\mathcal I_n^{-1}
-
\mathcal I_{n, N_n}^{-1}
\right)\,u,
\end{align}
where \(\mathcal I_{n,0}=\mathcal I_n\). Taking expectations, applying the tower property and multiplying by $n$ gives
\begin{align}
n\,s_{n}^2
&=
u^\top\,
\left(n\,\mathcal I_n^{-1}\right)\,u
-
n\,E\left(u^\top\,\mathcal I_{n,N_n}^{-1}\,u\right).
\label{eq:truncated_sn_decomp}
\end{align}
We now study the two terms on the right-hand side separately. First, by Assumption~\ref{as:In_conv_det},
\[
u^\top\,\left(n\,\mathcal I_n^{-1}\right)\,u
\to
u^\top\,\mathcal{I}_{P_X^*}^{-1}\,u.
\]
It thus remains to show that the truncation remainder is negligible. First, we note that
\begin{align*}
    \|N_n\,\mathcal I_{n,N_n}^{-1}\|_2 = \frac{N_n}{n + N_n + 1}\,  R_{n, N_n} \leq \sup_{i \geq 0} R_{ni}
\end{align*}
where $R_{ni}$ is defined in (\ref{eq:Rni}). From Lemma \ref{lem:sup_Ini}, we have
\[
\sup_{n \geq 1}
\|N_n\,\mathcal I_{n,N_n}^{-1}\|_2
<\infty
\quad \textnormal{a.s.}
\]
Consequently, since \(n/N_n\to0\), we have
\[
\left|n \,u^\top\,\mathcal I_{n,N_n}^{-1}\,u \right|
\le
\frac n{N_n}\,\|u\|^2
\|N_n\,\mathcal I_{n,N_n}^{-1}\|_2
\to0
\quad \textnormal{a.s.}
\]
which implies $n \,u^\top\,\mathcal I_{n,N_n}^{-1}\,u \to 0$ a.s.

To handle the expectation, we once again invoke dominated convergence. Note that for all $N_n \geq 1$, the matrix $\mathcal{I}_{n, N_n} = \mathcal{I}_n + \Sigma_{n, N_n}$ satisfies $\mathcal{I}_{n, N_n} \succeq \mathcal{I}_n$ a.s. 
Consequently, for any vector $u$, the quadratic form is bounded by
$$n\, u^\top\, \mathcal{I}_{n, N_n}^{-1}\, u \leq u^\top\, \left(n\, \mathcal{I}_n^{-1}\right)\, u \quad \text{a.s.}$$
By Assumption \ref{as:In_conv_det}, $n\, \mathcal{I}_n^{-1} \to \mathcal{I}_{P_X^*}^{-1}$ is a deterministic convergent sequence, meaning there exists a dominating constant for $u^\top\, (n\, \mathcal{I}_n^{-1})\, u$, thus giving
$$E\left( n\, u^\top\, \mathcal{I}_{n, N_n}^{-1}\, u \right) \to 0.$$
Combining this with the deterministic limit of the first term in \eqref{eq:truncated_sn_decomp}, we conclude that:
$$n\, s_{n}^2 \to u^\top\, \mathcal{I}_{P_X^*}^{-1}\, u.$$

\end{proof}

We are now ready to verify the conditional variance condition. 
\begin{lemma}\label{lem:condit_var}
     Under Assumptions \ref{as:suff_ident}, \ref{as:In_conv_det} and either \ref{as:lyapunov_easy} or \ref{as:lyapunov_hard}, we have
    \begin{align*}
        s_{n}^{-2}\, V_{n}^2 \to 1 \quad \textnormal{in probability.}
    \end{align*}
\end{lemma}
\begin{proof}
We highlight that while $V_{n}^2$ is challenging to work with, (\ref{eq:telescop_var}) is easy to work with.
We will thus connect the enlarged conditioning \(\left(\mathcal F_{n,i-1}, X_i\right)\) to $V_{n}^2$ using Proposition 6.16 of \citet{hausler2015stable}, which leverages the conditional Lindeberg condition. 
Consider the array of $\sigma$-algebras $\left(\mathcal{F}_{ni}\right)_{1 \leq i \leq N_n, n \geq 1}$ and random variables $\left(U_{ni}\right)_{1 \leq i \leq N_n, n\geq 1}$ where
\[
U_{ni}
=
\sqrt{
s_{n}^{-2}\,
E\left(Z_{ni}^2\mid \mathcal F_{n,i-1},X_i\right)
},
\]
and 
\[
\sum_{i=1}^{N_n}U_{ni}^2
=
\frac{
n\,u^\top\,\mathcal I_n^{-1}\,u
-
n\,u^\top\,\mathcal I_{n,N_n}^{-1}\,u
}{
n\,s_{n}^2
}\cdot
\]
The numerator can be handled via Lemma \ref{lem:s_nNn} minus the final expectation step, so we clearly have
\[
\sum_{i=1}^{N_n}U_{ni}^2
\to 1 \quad \text{a.s.} 
\]

It therefore remains to verify the required regularity conditions for $U_{ni}$. We first verify the square-integrability and that the sum of the conditional variances is bounded in probability. Taking expectations and applying the tower property, we have
\begin{align*}
\sum_{i = 1}^{N_n}E\left( U_{ni}^2\right) = 
E\left[
\sum_{i=1}^{N_n}\,
E\left(
U_{ni}^2
\mid
\mathcal F_{n,i-1}
\right)
\right]
&=
s_{n}^{-2}\,
\sum_{i=1}^{N_n}
E\left(
Z_{ni}^2
\right)
\\
&=
1,
\end{align*}
where the last equality follows from the definition of \(s_{n}^2\). Since each summand is nonnegative, each \(U_{ni}\) is clearly square-integrable. Moreover, from the above, the sequence
\[
\sum_{i=1}^{N_n}
E\left(
U_{ni}^2
\mid
\mathcal F_{n,i-1}
\right)
\]
for $n \geq 1$ is bounded in \(L^1\), so it follows from Markov's inequality that 
\[
\sum_{i=1}^{N_n}
E\left(
U_{ni}^2
\mid
\mathcal F_{n,i-1}
\right)
=
O_P(1).
\]

We next verify the conditional Lindeberg condition for $U_{ni}$. A similar argument to Lemma \ref{lem:condit_lyapunov} yields
\begin{align*}
\sum_{i=1}^{N_n}
E\left\{
U_{ni}^2\,\mathbbm 1 \left({|U_{ni}|>\varepsilon}\right) \mid \mathcal F_{n,i-1}\right\} \leq \varepsilon^{-2}\, \sum_{i=1}^{N_n} E\left(U_{ni}^4 \mid \mathcal F_{n,i-1} \right) \quad \text{a.s.}
\end{align*}
where
\[
U_{ni}^4=s_{n}^{-4}\,\left\{E\left(Z_{ni}^2\mid \mathcal F_{n,i-1},X_i\right)\right\}^2.
\]
Using Jensen's inequality,
\[
\left\{
E\left(Z_{ni}^2\mid \mathcal F_{n,i-1},X_i\right)
\right\}^2
\leq
E\left(Z_{ni}^4\mid \mathcal F_{n,i-1},X_i\right) \quad \text{a.s.}
\]
So
\begin{align*}
\sum_{i=1}^{N_n} E\left(U_{ni}^4 \mid \mathcal F_{n,i-1} \right) \leq
\,s_{n}^{-4}\,
\sum_{i=1}^{N_n}
E(Z_{ni}^4\mid \mathcal F_{n,i-1}) \quad \text{a.s.}
\end{align*}
Since this is exactly the conditional Lyapunov term for $\delta =2$, Lemma \ref{lem:pmp_lyapunov} immediately gives for all $\varepsilon > 0$
\[
\sum_{i=1}^{N_n}
E\left\{
U_{ni}^2\,
\mathbbm 1\left(|U_{ni}|>\varepsilon\right)
\mid\mathcal F_{n,i-1}
\right\}
\to0 \quad \text{in probability}.
\]
Proposition 6.16 of \citet{hausler2015stable} therefore yields 
\[
\sum_{i=1}^{N_n}
E\left(
U_{ni}^2
\mid
\mathcal F_{n,i-1}
\right)
-
\sum_{i=1}^{N_n}
U_{ni}^2
\to0
\quad\text{in probability}.
\]
Since $\sum_{i=1}^{N_n}
U_{ni}^2
\to1$ a.s., we have
\[
\sum_{i=1}^{N_n}
E\left(
U_{ni}^2
\mid
\mathcal F_{n,i-1}
\right)
\to1
\quad\text{in probability}.
\]
Substituting the definition of $U_{ni}$, we obtain
\[
s_{n}^{-2}\,
\sum_{i=1}^{N_n}
E\left(
Z_{ni}^2
\mid
\mathcal F_{n,i-1}
\right)
\to1
\quad\text{in probability}
\]
as desired.
\end{proof}

\subsection{Infinite martingale array and a.s. convergence}\label{sec:trunc}
Given Lemmas \ref{lem:pmp_lyapunov} and \ref{lem:condit_var}, we can now apply Theorem \ref{th:mart_CLT} to obtain
\begin{align*}
s^{-1}_{n}\, u^\top\left(\beta_{n,N_n} - \beta_n\right) \to \mathcal{N}(0,1)
\end{align*}
in distribution.  Lemma \ref{lem:s_nNn}, Slutsky's theorem and the Cram\'{e}r-Wold theorem then give
\begin{align*}
n^{1/2}\left(\beta_{n,N_n} - \beta_n\right) \to \mathcal{N}\left(0,\mathcal{I}_{P_X^*}^{-1}\right)
\end{align*}
in distribution. 

Next, we make the connection to the infinite martingale array $(\beta_{ni}, \mathcal{F}_{ni})_{i \geq 1, n \geq 1}$ based on a similar argument to the proof of \citet[Theorem 3.6]{Hall2014}. We first note that
\begin{align*}
    E\left\{n(\beta_{n\infty} - \beta_{n,N_n})(\beta_{n\infty} - \beta_{n,N_n})^\top \right\} = nE\left(\mathcal{I}_{n,N_n}^{-1}\right)
\end{align*}
following Theorem \ref{thm:predictive_invariance_mp}. We have already shown in the proof of Lemma \ref{lem:s_nNn} that $N_n/n  \to \infty$ implies
\begin{align*}
    nE\left(\mathcal{I}_{n,N_n}^{-1}\right) \to 0
\end{align*}
as $n \to \infty$. As a result, we clearly have
\begin{align*}
n^{1/2}(\beta_{n\infty} - \beta_{n,N_n}) \to 0 \quad \text{in probability}
\end{align*} 
as $n \to \infty$, which in turn implies
\begin{align*}
n^{1/2}\left(\beta_{n\infty} - \beta_n\right) \to \mathcal{N}\left(0,\mathcal{I}_{P_X^*}^{-1}\right)
\end{align*}
in distribution. 

Finally, the same argument as for \citet[Theorem 3]{fong2026asymptotics} can be applied to convert the above result to the weak convergence a.s. case under Assumption \ref{as:In_conv}, thereby giving Theorem \ref{thm:as_invar}.

\subsection{Verifying assumptions}\label{sec:assumptions}
In this subsection, we briefly provide some discussion on how to verify Assumption \ref{as:In_conv} in practice, as well as some 
 examples of models satisfying Assumptions \ref{as:lyapunov_easy} and \ref{as:lyapunov_hard}. 

For Assumption \ref{as:In_conv}, we note that the latter condition on the a.s. convergence of $n^{-1}\mathcal{I}_n$ can be shown under the assumption of $\beta_n \to \beta^*$ a.s. and appropriate conditions on $w(t)$ and $P_X^*$ by using a uniform law of large numbers argument. Sufficient conditions are the variants of Assumptions \ref{as:suff_ident} and either \ref{as:lyapunov_easy} or \ref{as:lyapunov_hard} where $P_X$ is replaced with $P_X^*$. It is likely that weaker assumptions will work, but we consider these for cohesiveness. We now provide a rough outline of the argument. The constant case $w(t) = w_0$ is obvious as we can simply rely on the strong law of large numbers. For the latter case, let us define
\begin{align*}
    \mathcal{I}^*_n(\beta) = \frac{1}{n}\sum_{j = 1}^n w(\beta^\top X_j^*)\,  X_j^* X_j^{*\top}, \quad \mathcal{I}^*(\beta) = E_{P_X^*}\left\{w(\beta^\top X) XX^\top\right\}.
\end{align*}
To verify the conditions required for $\mathcal{I}_n^*(\beta_n) \to \mathcal{I}^*(\beta^*)$, we can use the decomposition
\begin{align*}
    \left\|\mathcal{I}_n^*(\beta_n) - \mathcal{I}^*(\beta^*) \right\| \leq  \left\|\mathcal{I}_n^*(\beta_n) - \mathcal{I}^*(\beta_n) \right\| +  \left\|\mathcal{I}^*(\beta_n) - \mathcal{I}^*(\beta^*) \right\|. 
\end{align*}
Note that since $0 < w(t) < \infty$ is continuous, $P_X^*$ has compact support, and $\beta_n \to \beta^*$ $P^{*\infty}$-a.s., we have that $0 < c_l \leq w(\beta_n^\top X_j^*) \leq c_u < \infty$ $P^{*\infty}$-a.s. from Lemma \ref{lem:compact_wt}. As a result, $\mathcal{I}^*(\cdot)$ is continuous from continuity and dominance of $w(t)$, so the second term goes to 0. The first term can be handled by a uniform strong law of large numbers, where
\begin{align*}
    \sup_{\beta \in B(\beta^*)}\left\|\mathcal{I}_n^*(\beta) - \mathcal{I}^*(\beta) \right\|\to 0 \quad P^{*\infty}\text{-a.s.}
\end{align*}
for $B(\beta^*)$ as a suitable compact neighborhood of $\beta^*$, following the convergence of $\beta_n$ to $\beta^*$. The conditions for the uniform strong law again hold under similar continuity and dominance arguments of $w(t)$.

For Assumptions \ref{as:lyapunov_easy} and \ref{as:lyapunov_hard}, we include a few common examples to illustrate the forms of $w(t)$ and $k(t)$ corresponding to the two cases.
\begin{example}[Location model]
    For location models, consider $p_Y(y \mid t) = p_0(y - t)$, which encompasses many examples including the Student-$t$ example of Section \ref{sec:sim}. Under appropriate regularity conditions, we have
    \begin{align*}
        r(y \mid t) = -\frac{p_0'(y - t)}{p_0(y-t)}
    \end{align*}
    which gives
    \begin{align*}
        w(t) = \int \left\{\frac{p_0'(y - t)}{p_0(y-t)}\right\}^2 p_Y(y \mid t) \, dy = \int \left\{\frac{p_0'(z)}{p_0(z)}\right\}^2 p_0(z)\, dz 
    \end{align*}
    through a change of variables $z = y-t$. Since the above does not depend on $t$, we clearly have $w(t) = w_0$ as desired. A similar argument can be applied for the kurtosis function to show $k(t) = k_0$. We thus satisfy Assumption \ref{as:lyapunov_easy}.
\end{example}

\begin{example}[Gamma with log link]
    For the Gamma family with a log link, we have $P_Y(\cdot \mid t) = \text{Gamma}\left\{\alpha, \alpha\exp(-t)\right\}$, where we are using the shape/rate parametrization. This then gives
\begin{align*}
    r(y \mid t) = \alpha \left\{y\exp(-t) - 1\right\},
\end{align*}
and one can verify that
\begin{align*}
    w(t) = \alpha, \quad k(t) = 3 + \frac{6}{\alpha}
\end{align*}
which are both constant for fixed shape, so we can apply Assumption \ref{as:lyapunov_easy}.
\end{example}

\begin{example}[Logistic]\label{ex:logistic}
For logistic regression, we can verify that
\begin{align*}
    w(t) = \frac{\exp(t)}{\left\{1 + \exp(t)\right\}^2}, \quad k(t) = \frac{1}{w(t)} - 3
\end{align*}
which is clearly not constant, and $k(t)$ can be very large for large $|t|$. This intuitively explains the requirement of $P_X$ having compact support in Assumption \ref{as:lyapunov_hard}.
\end{example}

\begin{example}[Poisson with log link]
For the Poisson family with a log link, one can show that
\begin{align*}
    w(t) = \exp(t), \quad k(t) = \frac{1}{w(t)} + 3
\end{align*}
where Assumption \ref{as:lyapunov_hard} is again relevant.
\end{example}

\section{Fixed-design case}\label{sec:fixed_design}
As mentioned in the main paper, it is more computationally efficient to precompute required quantities such as $\mathcal{I}_{i}^{-1} X_i$  for a single realization of the predictive covariate sequence $X_{n+1:N}$ for a sufficiently large truncation iteration $N$, which is then used across all parallel predictive resampling trajectories. Further computational details are provided in Section \ref{sec:computation}.

The above setting is closely related to the fixed-design formulation, where the predictive covariate sequence $(X_i)_{i \ge n+1}$ is treated as a fixed deterministic sequence instead of being random samples from $P_X$. 
In this case, the filtration is $(\mathcal{F}_i)_{i \geq n+1}$ where $\mathcal{F}_i = \sigma(Y_{n+1},\dots,Y_i)$. We again use the notation  $\mathcal{F}_n = \{\emptyset,\Omega\}$ to indicate the observed dataset $\left(X_{1:n}, Y_{1:n}\right)$ is treated as constant. We now show that under standard regularity conditions, this fixed-design setting preserves weak design invariance, where the proof is very similar to that of Theorem \ref{thm:predictive_invariance_mp}. We highlight that the posterior covariance in the fixed-design case is also exactly the same as the random-design case as shown in the following result, so no uncertainty is lost in this setting. We begin with the fixed-design version of Assumption \ref{as:suff_ident}.
\begin{assumption}\label{as:fix_suff_ident}
    Let $\tilde{X} = (\tilde{X}_i)_{i \ge 1}$ be a deterministic covariate sequence which satisfies
\begin{align*}
\lim_{N \to \infty} \frac{1}{N}\, \sum_{i=1}^N  \tilde{X}_i \tilde{X}_i^\top = \Sigma_{X},
\label{eq:fixed_design_lln}
\end{align*}
where $\Sigma_X$ is finite and $\Sigma_X \succ 0$. Furthermore, we require
\begin{align*}
    \limsup_{N \to \infty }\, \frac{1}{N}\sum_{i = 1}^N \|\tilde{X}_i\|^4 < \infty.
\end{align*}
\end{assumption}
The extra condition is required to control the tails of the deterministic sequence for the non-constant $w(t)$ case. We now consider predictive resampling with update rule (\ref{eq:recur_update}) and $P_i(\cdot \mid x) = P_{Y}(\cdot  \mid \beta_i^\top x)$, where we set $X_i = \tilde{X}_{i-n}$ for $i \geq n+1$ as our future covariate sequence. 
\begin{corollary}[Weak design invariance, fixed-design]
\label{cor:fixed_design}
 Under predictive resampling as described above, $(\beta_i, \mathcal{F}_i)_{i \geq n+1}$ is a martingale bounded in $L^2$ and $\beta_i \to \beta_\infty$ almost surely, so $E(\beta_\infty \mid \mathcal{F}_n)$ does not depend on $\tilde{X}$.
 If Assumption \ref{as:fix_suff_ident} holds and $w(t)$ is continuous and strictly positive for all $t \in \mathbb{R}$, then $\textnormal{Cov}(\beta_\infty \mid \mathcal{F}_n)$ also does not depend on $\tilde{X}$. 
\end{corollary}

\begin{proof}
The update $Z_i = \beta_i - \beta_{i-1}$ satisfies $E(Z_i \mid \mathcal{F}_{i-1}) = 0$ since $E\{s(\beta_{i-1}, Y_i \mid X_i)\mid \mathcal{F}_{i-1}\} = 0$. Thus, $(\beta_i, \mathcal{F}_i)_{i \ge n+1}$ is a martingale.  Following the identity in Lemma \ref{lem_increment_cov}, the conditional covariance of the increments is
$$E(Z_i Z_i^\top \mid \mathcal{F}_{i-1}) = \mathcal{I}_{i-1}^{-1} - \mathcal{I}_i^{-1}.$$ 
We can thus also show a similar result to Lemma \ref{lem:post_cov} albeit without the expectation, that is
    \begin{align*}
       \textnormal{Cov}(\beta_N \mid \mathcal{F}_n)
&= \mathcal{I}_n^{-1} - \mathcal{I}_N^{-1}.
    \end{align*}

To show $\mathcal{I}_N^{-1} \to 0$, we can follow a similar proof technique to Lemma \ref{lem:EIn_0}, although more work is needed to show the existence of a sufficiently large compact ball $B_r$ to control the minimum eigenvalue of $N^{-1}\mathcal{I}_N$. Again for $B_r = \{x \in \mathbb{R}^p: \|x\| \leq r\}$, let us define
\begin{align*}
 \Sigma_N(r) = \frac{1}{N}\sum_{i=1}^N \tilde X_i\tilde X_i^\top\mathbbm 1(\tilde X_i\in B_r)
\end{align*}
with $\Sigma_N := \Sigma_N(\infty)$. We must now show the existence of a sufficiently large $r$ such that
\begin{align*}
\liminf_{N\to \infty}\lambda_{\min}\left\{   \Sigma_N(r)\right\}>0.
\end{align*}
To do this, Weyl's inequality gives
\begin{align*}
\lambda_{\min}\left\{   \Sigma_N(r)\right\} \geq \lambda_{\min}\left(   \Sigma_N\right) - \lambda_{\max}\left\{ \frac{1}{N}\sum_{i=1}^N \tilde X_i\tilde X_i^\top\mathbbm 1(\tilde X_i\in B_r^c)\right\}.
\end{align*}
Since $\lambda_{\min}(\Sigma_N) \to \lambda_{\min}\left(\Sigma_X\right) > 0$ from Assumption \ref{as:fix_suff_ident}, we only need to upper bound the latter term. Using a Rayleigh quotient argument, one can show that
\begin{align*}
    \lambda_{\max}\left\{ \frac{1}{N}\sum_{i=1}^N \tilde X_i\tilde X_i^\top\mathbbm 1(\tilde X_i\in B_r^c)\right\} \leq \frac{1}{N}\sum_{i = 1}^N \| \tilde{X}_i\|^2 \mathbbm{1}\left(\|\tilde{X}_i\| > r\right).
\end{align*}
Finally, a Lyapunov style argument gives
\begin{align*}
     \frac{1}{N}\sum_{i = 1}^N \| \tilde{X}_i\|^2 \mathbbm{1}\left(\|\tilde{X}_i\| > r\right) \leq \frac{1}{r^2}\left(  \frac{1}{N}\sum_{i = 1}^N \| \tilde{X}_i\|^4\right).
\end{align*}
From Assumption \ref{as:fix_suff_ident}, let
\[
C=
\limsup_{N\to\infty}
\frac1N
\sum_{i=1}^N
\|\tilde X_i\|^4
<\infty.
\]
We can then choose $r$ sufficiently large so that
\[
\frac{C}{r^2}
<
\frac12\lambda_{\min}(\Sigma_X),
\]
which gives
\[
\liminf_{N\to\infty}
\lambda_{\min}\{\Sigma_N(r)\}
\ge
\frac12\lambda_{\min}(\Sigma_X)
>
0.
\]
The rest of the proof then proceeds as in Lemma \ref{lem:EIn_0}, where we again obtain the lower bound $$\liminf_N \lambda_{\min}\left(N^{-1} \mathcal{I}_N\right) > 0.$$
Finally, we can skip the dominance argument, and $\mathcal{I}_N^{-1} \to 0$ implies $\textnormal{Cov}(\beta_\infty \mid \mathcal{F}_n) =\mathcal{I}_n^{-1}$ as in the random-design case.
\end{proof}

Similarly, it is possible to show a version of Theorem \ref{thm:as_invar} under the case where the predictive design sequence (and potentially the observed covariate sequence) is deterministic. Again, the proof will generally be easier as we do not need to take expectations over the covariate distribution. To begin, we note that Assumptions  \ref{as:lyapunov_easy} and \ref{as:lyapunov_hard} can be easily replaced by their fixed-design counterparts, given below. Since Assumption \ref{as:fix_suff_ident} already contains the 4th moment control, it will not be necessary for Assumption \ref{as:fixed_lyapunov_easy}.
\renewcommand{\theassumption}{\thesection.\arabic{assumption}(i)}
\begin{assumption}\label{as:fixed_lyapunov_easy}
Suppose $w(t) = w_0$ and  
$k(t) =k_0$ where $w_0$ and $k_0$ are positive finite constants. 
\end{assumption}
\addtocounter{assumption}{-1}
\renewcommand{\theassumption}{\thesection.\arabic{assumption}(ii)}
\begin{assumption}\label{as:fixed_lyapunov_hard}
Suppose $w(t)$ and $k(t)$ are continuous and strictly positive for all $t \in \mathbb{R}$. Let $\tilde{X} = (\tilde{X}_i)_{i \ge 1}$ satisfy $\sup_{i\geq 1} \|\tilde{X}_i\| < \infty$. 
\end{assumption}
\renewcommand{\theassumption}{\thesection.\arabic{assumption}}
Note that Assumption \ref{as:In_conv} and \ref{as:In_conv_det} can remain unchanged. The real observations may be random-design or fixed-design,  where in the latter case we have a real design sequence $\tilde{X}^* = (\tilde{X}_i^*)_{i \geq 1}$ and  $P^{*\infty}$ is now the joint law of $Y_{1:\infty}^*$ which is no longer identically distributed. Analogous to Corollary \ref{cor:fixed_design}, the filtration we use will be $(\mathcal{F}_{ni})_{i \geq 1, n \geq 1}$ where $\mathcal{F}_{ni} = \sigma\left(Y_{n1},\ldots,Y_{ni}\right)$ and we use the same deterministic sequence $(\tilde{X}_i)_{i \geq 1}$ for each row in the martingale array. We then have the following result. 
\begin{corollary}[Asymptotic normality, fixed-design]\label{cor:norm}
    Suppose Assumptions \ref{as:fix_suff_ident}, \ref{as:In_conv} and either \ref{as:fixed_lyapunov_easy} or \ref{as:fixed_lyapunov_hard} hold. Let $\beta_{n\infty}$ be the limit under predictive resampling initialized at $\beta_n$ as in Corollary \ref{cor:fixed_design}. Then the posterior law of $n^{1/2}\,(\beta_{n\infty} - \beta_n)$ converges weakly $P^{*\infty}$-almost surely to a Gaussian as $n \to \infty$.
\end{corollary}
\begin{proof}
    The proof structure is largely unchanged, with the fixed-design case actually simplifying much of the proof. For example, Lemma \ref{lem:sup_Ini} will hold deterministically instead of almost surely, Lemma \ref{lem:s_nNn} and Section \ref{sec:trunc} will not require an expectation over $\mathcal{I}_{n,N_n}^{-1}$, and Lemma \ref{lem:condit_var} is no longer necessary since $V_{n}^2 = s_{n}^{2}$.\vspace{1mm}
    
    One difference is the proof of Lemma \ref{lem:pmp_lyapunov} in the constant $w(t) = w_0$ and $k(t)= k_0$ case. As we do not take an expectation over $\|\tilde{X}_i\|^4$, we must handle   this term in (\ref{eq:gen_EZ}) differently. Instead, we will have 
    \begin{align*}
E\left(\widetilde{Z}_{ni}^4 \mid \mathcal{F}_{n,i-1}\right) \leq w_0^2\, k_0\, K_3 \, \|u\|^4\, \|\tilde{X}_i\|^4 
    \end{align*}
    where $K_3$ is the deterministic upper bound on $\sup_{n \geq 1, i \geq 0} R_{ni}^4$ from the fixed-design version of Lemma \ref{lem:sup_Ini}. To check the conditional Lyapunov condition with $\delta = 2$, we then note that
    \begin{align*}
        s_{n}^{-4} \sum_{i = 1}^{N_n}E\left({Z}_{ni}^4 \mid \mathcal{F}_{n,i-1}\right)  \leq O(n^2) \sum_{i = 1}^{\infty}(n+i)^{-4} \|\tilde{X}_i\|^4.
    \end{align*}
To control the above, we have
\begin{align*}
    \sum_{i = 1}^\infty (n+i)^{-4} \|\tilde{X}_i\|^4 = \sum_{i = 1}^\infty \|\tilde{X}_i\|^4 \int_{n+i}^\infty 4s^{-5} ds\,   &= \int_{n+1}^\infty 4s^{-5} \left(\sum_{i = 1}^{\lfloor{s - n}\rfloor} \|\tilde{X}_i\|^4\right)\,ds.
\end{align*}
    As Assumption \ref{as:fix_suff_ident} gives $\sum_{i = 1}^N \|\tilde{X}_i\|^4 = O(N)$, we have
    $$\sum_{i = 1}^{\infty}(n+i)^{-4} \|\tilde{X}_i\|^4 = O(n^{-3}),$$ 
    and the conditional Lyapunov condition holds. Finally, we note that the general $w(t)$ and $k(t)$ case is unchanged as we rely on the boundedness of $\tilde{X}_i$ here.
\end{proof}

Finally, we highlight that it would be interesting to extend Doob's consistency theorem to the fixed-design case, which is likely possible in our simplified linear setting, although we leave a proper investigation for future work. The fixed-design setting is for example, addressed in \citet[Theorem 7.7]{choi2008remarks} for the linear regression model with nonparametric errors, and a potential direction is to investigate if \citet[Proposition 7.5]{choi2008remarks} can be applied to our setting.

\section{Identifiability}\label{sec:ident}
\subsection{Limit under non-identifiability}\label{sec:non_ident}
In this section, we illustrate how predictive resampling can still converge under violation of Assumption \ref{as:suff_ident}, albeit to an incorrect limit. We consider the general parametric martingale posterior setup given in (\ref{eq:recur_update}), as this also encompasses the Bayesian linear regression case given in Example \ref{ex1}.

We begin by noting that both Lemmas \ref{lem:post_cov} and \ref{lem_L2} hold even if $P_X$ does not satisfy Assumption \ref{as:suff_ident}, as this property is not used in those results. As a result, $(\beta_i,\mathcal{F}_i)_{i \geq n+1}$ is still a martingale bounded in $L^2$, so $\beta_i$ still converges to $\beta_\infty$ a.s. under non-identifiability. {However, the key issue is that the underlying parameter $\beta$ is not measurable with respect to $\sigma(X_{1:\infty}, Y_{1:\infty})$ under non-identifiability. Consequently, the limit $\beta_\infty$ fails to recover $\beta$ as required in Doob's consistency theorem, thus invalidating the uncertainty quantification obtained from predictive resampling.} We see the effects of this in Example \ref{ex1}, where we obtain a mismatched posterior to the ground truth, and discuss this further in Section \ref{sec:pmp_ident}.

It is also possible to characterize this mismatch. The identifiability condition only plays a role in the final step of the proof of Theorem  \ref{thm:predictive_invariance_mp}. Although we will still have
    \begin{align*}
E\!\left\{
\left(\beta_N-\beta_n\right)\,\left(\beta_N-\beta_n\right)^\top
\mid
\mathcal F_n,X_{n+1:N}
\right\}
&= \mathcal{I}_n^{-1} - \mathcal{I}_N^{-1},
\end{align*}
under non-identifiability, we do not generally have $\mathcal{I}_N^{-1} \to 0$ a.s. A key side effect of violating Assumption \ref{as:suff_ident} is that the distribution of $\beta_\infty$ will no longer be weakly design invariant.

As an example, consider the linear regression model with $\sigma = 1$ and a ridge penalty, where 
\begin{align*}
    \mathcal{I}_N = \sum_{j=1}^N X_j X_j^\top
+\tau I,
\end{align*}
and suppose as well that $p > n$ and $P_X = \mathbb{P}_X$ which makes the analysis easier. This can be shown to be equivalent to Example \ref{ex1}; see Section \ref{sec:update_motiv}. We will still have 
$$\frac{1}{N-n}\sum_{j = n+1}^N X_j X_j^\top \to \Sigma_n \:=\frac{1}{n}\sum_{i = 1}^n X_i X_i^{\top} \quad \text{a.s.}$$ 
but $\Sigma_n$ is no longer invertible, violating Assumption \ref{as:suff_ident}. Next, we write
\begin{align*}
    \mathcal{I}_N \approx N \Sigma_n + \tau I,
\end{align*}
where we have used the fact that $\mathcal{I}_n = n \Sigma_n + \tau I$. Next, we note that
\begin{align*}
    \mathcal{I}_N^{-1} &\approx \tau^{-1}\delta_N\left( \Sigma_n + \delta_N I\right)^{-1} = \tau^{-1} \left\{ I - \left(\Sigma_n + \delta_NI\right)^{-1}\Sigma_n \right\}
\end{align*}
where $\delta_N = \tau/N$. Taking $\delta_N \to 0$ then gives
\begin{align*}
    \mathcal{I}_N^{-1} \to \tau^{-1}\left(I - \Sigma_n^{\dagger}\Sigma_n\right) \quad \text{a.s.}
\end{align*}
where $\Sigma_n^{\dagger}$ is the Moore-Penrose inverse of $\Sigma_n$. To see the above, we can take the spectral decomposition $\Sigma_n = U \Lambda U^\top$, which gives
\begin{align*}
    \lim_{\delta \to 0}\left(\Sigma_n + \delta I\right)^{-1}\Sigma_n= U \Lambda^0 U^\top = \Sigma_n^{\dagger} \Sigma_n,
\end{align*}
where $\Lambda^0_{jj} = \mathbbm{1}(\Lambda_{jj} > 0)$. If $\Sigma_n$ were invertible, then the limit of $\mathcal{I}_N^{-1}$ is clearly 0. Otherwise, since we have $\mathcal{I}_n^{-1} = (n\Sigma_n + \tau I)^{-1}$, the limiting covariance under predictive resampling takes the form
\begin{align*}
    \left(n\Sigma_n + \tau I\right)^{-1} - \tau^{-1}\left(I - \Sigma_n^{\dagger}\Sigma_n\right) &= \left(n\Sigma_n + \tau I\right)^{-1}\left\{I - \tau^{-1}\left(n\Sigma_n + \tau I\right)\left(I - \Sigma_n^{\dagger}\Sigma_n\right) \right\}\\
    &= \left(n\Sigma_n + \tau I\right)^{-1}\Sigma_n^{\dagger} \Sigma_n.
\end{align*}
One can verify that the above is still a valid covariance matrix, and since $(I - \Sigma_n^{\dagger}\Sigma_n) \succeq 0$,
the above is clearly smaller in Loewner order than $\mathcal{I}_n^{-1} =  \left(n\Sigma_n + \tau I\right)^{-1}$. As a result, all marginal variances will be smaller, which agrees with Figure \ref{fig:pr_vs_bayes_sim_lm} (Left) under $P_X = \mathbb{P}_X$. 

Finally, we repeat the experiment from Section \ref{sec:sim} but with $P_X = \mathbb{P}_X$, which is illustrated in Figure \ref{fig:pr_vs_bayes_sim_px}. As we can see, the marginal variances are again smaller than that of Bayes or the parametric martingale posterior with an identifiable $P_X$, especially for the inactive component. This further demonstrates the importance of identifiability for weak design invariance.
\begin{figure}[h!]
\centering
\includegraphics[width=0.96\linewidth]{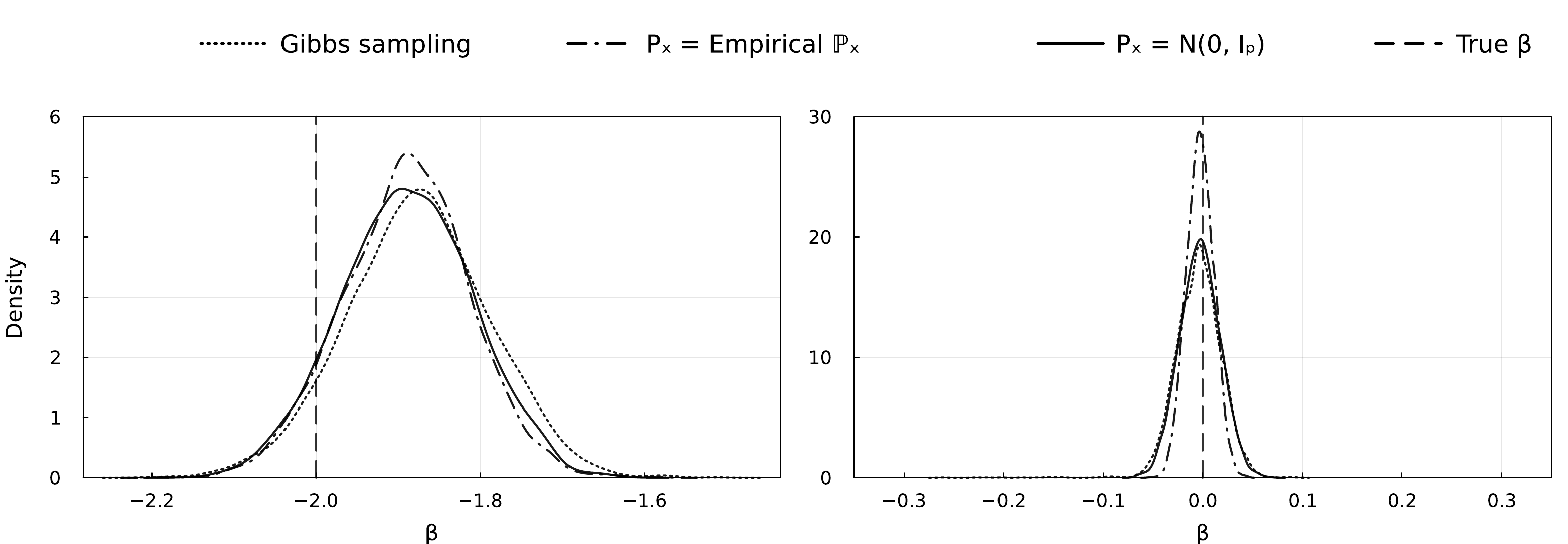}
\caption{Posteriors for the Student-$t$ example for selected (Left) active coefficient; (Right) inactive coefficient. Predictive resampling is with update rule \eqref{eq:recur_update} using $P_X = \mathbb{P}_X$ versus $P_X = \mathcal N(0, I_{p})$. } 
\label{fig:pr_vs_bayes_sim_px}
\end{figure}

\subsection{Identifiability for the parametric martingale posterior}\label{sec:pmp_ident}
In the main paper, we \textit{define} the distribution of $\beta_\infty$ to be the martingale posterior, and proceed to verify properties such as existence, weak design invariance and asymptotic normality. While the practical impact of choosing a non-identifiable $P_X$ is quite visible on the parametric martingale posterior (e.g. as discussed in Section \ref{sec:non_ident}), the role of identifiability in general is less clear than in traditional Bayes. For example, Theorems \ref{thm:predictive_invariance_mp} and \ref{thm:as_invar} do not explicitly rely on the one-parameter family $P(\cdot \mid t)$ being identifiable, although it is intuitively meaningful. We provide an informal outline of the importance of identifiability here. 

As in the proof of Proposition \ref{prop:doob_regression}, let $P_\beta$ denote the joint distribution of $(X,Y)$ where we compose $P_\beta(\cdot \mid x)$ with the fixed $P_X$. One way to highlight the role of identifiability is to consider the functional view of the nonparametric martingale posterior \citep{fong2023martingale}. Since $\beta_i \to \beta_\infty$ a.s. under predictive resampling, we will have $P_\beta \to P_{\beta_\infty}$ weakly a.s. under appropriate regularity conditions. Note that the empirical distribution $\mathbb{P}_i$ of $(X_i,Y_i)$ will also have the same weak a.s. limit, e.g. \citet[Corollary 2]{fong2024bayesian} or \citet[Theorem 2.2]{berti2004limit}. The functional view would then regard the distribution of functionals of $P_{\beta_\infty}$ as the martingale posterior.
Under the identifiability assumptions of Proposition \ref{prop:identifiability_sufficient}, the mapping $\beta \to P_\beta$ is injective, so there will exist a functional $f$ such that $\beta_\infty = f(P_{\beta_\infty})$ a.s. This can likely be formalized (e.g. measurability of $f$) using a similar argument as in Doob's theorem \citep[Theorems 4.1, 4.2]{miller2018detailed}. The interesting property is that the distribution of $\beta_\infty$ can still exist even if $P_Y( \cdot \mid t)$ is not identifiable, or if $P_X$ does not satisfy Assumption \ref{as:suff_ident}, although we lose the functional interpretation and arguably the validity of the resulting posterior.

\section{Motivation for recursive update}\label{sec:update_motiv}
In this section, we motivate our recursive update form, in particular the choice of $\mathcal{I}_i$ and $\alpha_i(x)$, through looking at Bayesian linear regression. Let $P_\beta(\cdot \mid x) = \mathcal{N}(\beta^\top x, \sigma^2)$ with a conjugate prior $\beta \sim \mathcal{N}(0,\tau^{-1} I)$. It is well known that the posterior mean after $i$ observations takes the form
\begin{align*}
    \beta_i = \left(\sum_{j=1}^{i} \frac{X_j\,X_j^{\top}}{\sigma^2} +\tau I \right)^{-1}\,\sum_{j=1}^{i}\frac{X_j\,Y_j}{\sigma^2}\cdot
\end{align*}
A well known identity for recursive least squares is
\begin{align*}
    \beta_i = \beta_{i-1} + \left(\sum_{j=1}^{i} \frac{X_j\,X_j^{\top}}{\sigma^2} +\tau I \right)^{-1}\frac{X_i\left(Y_i - \beta_{i-1}^\top X_i\right)}{\sigma^2}\cdot
\end{align*}
It is not too hard to verify that for the linear regression model, we have
\begin{align*}
  s(\beta,y \mid x) = \frac{x(y- \beta^\top x)}{\sigma^2}, \quad w(t) = w_0 = \frac{1}{\sigma^2}
\end{align*}
hence allowing us to write
\begin{align*}
     \beta_i &= \beta_{i-1} + \mathcal{I}_i^{-1} s(\beta_{i-1},Y_i \mid X_i),\\
     \mathcal{I}_i &= \sum_{j = 1}^i w_0 X_j X_j^\top + D
\end{align*}
where $D = \tau I$. To motivate the scalar correction term, we note that the posterior predictive implies
\begin{align*}
    \left\{\left(Y_i - \beta_{i-1}^\top X_i\right)\mid \mathcal{F}_{i-1}, X_{i}\right\} \sim \mathcal{N}\left\{0, \sigma^2 \left(1 + w_0 X_i^\top \mathcal{I}_{i-1}^{-1}X_i\right) \right\}.
\end{align*}
Now suppose instead $\{(Y_i - \beta_{i-1}^\top X_i)\mid \mathcal{F}_{i-1}, X_i\} \sim \mathcal{N}(0,\sigma^2)$ as we want to use the plug-in predictive. If we scale $(Y_i - \beta_{i-1}^\top X_i)$ by $\alpha_{i-1}(X_i)$, where
\begin{align*}
    \alpha_{i}(x) = \left(1 + w_0 \, x^\top \mathcal{I}_{i}^{-1}\,x\right)^{1/2},
\end{align*}
then the scaled residual would match the posterior predictive case. Combining this, we obtain the update
\begin{align*}
     \beta_i &= \beta_{i-1} + \alpha_{i-1}(X_i)\,\mathcal{I}_i^{-1} s(\beta_{i-1},Y_i \mid X_i).
\end{align*}
Finally, we note that unlike what one might expect for stochastic gradient descent (e.g. for the maximum a posteriori estimate), the prior regularization only enters into $\mathcal{I}_i$ without affecting the score function for Bayesian updating of the posterior mean.

\subsection{Generalizing beyond linear regression}\label{sec:gen_w}
To generalize beyond Gaussian linear regression, the score function can easily just be replaced with that of the appropriate density. For $\mathcal{I}_i$, we notice that the conditional Fisher information for the case $P_\beta(\cdot \mid x) = P_Y(\cdot \mid \beta^\top x)$ would take the form
\begin{align*}
    E\left\{s(\beta, Y \mid x)\, s(\beta, Y \mid x)^\top \mid x\right\} = w(\beta^\top x)\,  xx^\top,
\end{align*}
so it is natural to replace all occurrences of $w_0$ with $w(\beta^\top x)$. We opt to plug-in the value of the initial estimate $\beta_n$ here, as this is our best estimate of $\beta^*$, and do so similarly for $D(\beta_n)$ as we discuss in Section \ref{sec:D_motiv}. 
Although it will not make much difference in practice, we briefly consider a variant of the update in Section \ref{sec:adaptive_I} where $w(\cdot)$ may depend on updated values of $\beta_i$, and illustrate the technical challenges this poses. Finally, the role of the additional ratio term $w(\beta_n^\top x)/w(\beta_i^\top x)$ in $\alpha_i(x)$ can be explicitly seen in the proof of Theorem \ref{thm:predictive_invariance_mp} to attain design invariance.

\subsection{Generalizing the regularization matrix $D$}\label{sec:D_motiv}

In high-dimensional regression, the assumption of sparsity is often necessary, which is incorporated through shrinkage priors in the Bayesian setting. In this subsection, we utilize this connection to inspire generalized choices of $D$ beyond just $D = \tau I$ as motivated at the start of this section. Many classic shrinkage priors can be expressed as a scale mixture of normals \citep{west1987scale}, which takes the form
\[
\beta^{(j)} \mid \tau^{(j)}
\sim
\mathcal N(0,1/{\tau^{(j)}}), \quad \tau^{(j)}
\sim
p(\tau^{(j)}\mid \lambda),
\quad
j=1,\dots,p
\]
where $\lambda$ are hyperparameters. For the model $p_\beta(y \mid x) = \mathcal{N}\left(y; \beta^\top x, \sigma^2\right)$ as considered at the start of this section, this prior choice is conjugate conditional on $(\tau^{(1)},\ldots,\tau^{(p)})$, giving a conditional posterior mean of
\begin{align*}
    E\left\{\beta \mid \mathcal{F}_i, \left(\tau^{(1)},\ldots,\tau^{(p)} \right)\right\} =  \left(\sum_{j=1}^{i} \frac{X_j\,X_j^{\top}}{\sigma^2} +D\right)^{-1}\,\sum_{j=1}^{i}\frac{X_j\,Y_j}{\sigma^2}
\end{align*}
where $D = \text{Diag}\left(\tau^{(1)},\ldots,\tau^{(p)}\right)$. The above provides a clear justification for the form of $\mathcal{I}_i$, where the remaining question is what to plug in for $\left(\tau^{(1)},\ldots,\tau^{(p)}\right)$.

As discussed in Section \ref{sec:plugin}, predictive resampling is initialized at the maximum a posteriori estimate $\beta_n$ given $(X_{1:n}, Y_{1:n})$ corresponding to the chosen sparse prior on $\beta$. Given this initialization, we propose to set $\left(\tau^{(1)},\ldots,\tau^{(p)}\right)$ as a function of the initial estimate $\beta_n$ by capturing its estimated sparsity, thereby justifying the notation $D(\beta_n)$ used in the main paper. As with $w(\beta_n^\top x)$ in Section \ref{sec:gen_w}, we hold $D(\beta_n)$ fixed throughout predictive resampling. 

To derive appropriate plug-in estimates of $\left(\tau^{(1)},\ldots,\tau^{(p)}\right)$, we suggest leveraging the conditional posterior utilized in Gibbs samplers. One potential choice is the Gaussian conditional posterior mean
\begin{align*}
  D_{jj}(\beta_n) &=  E\left(\tau^{(j)} \mid \mathcal{F}_n, \beta = \beta_n \right) = E\left(\tau^{(j)} \mid \beta^{(j)} = \beta_n^{(j)} \right)
\end{align*}
 for $j = 1,\ldots,p$, where we highlight that the conditional posterior of $(\tau^{(1)},\ldots, \tau^{(p)})$ does not depend on the observations. The above term is often in closed form for scale mixture of normals priors; otherwise we may use the conditional maximum a posteriori estimate which can be computed using numerical techniques. Other plug-in estimates can be used if the above suggestions are not stable.
 
We now derive the plug-in estimates given in Table \ref{tab:D_n_summary} following the above construction. 
\begin{example}[LASSO]
 The well-known Bayesian LASSO prior \citep{park2008bayesian} is a scale mixture of normals prior with
 \begin{align*}
    \frac{1}{\tau^{(j)}} \sim \text{Exp}\left(\frac{\lambda^2}{2}\right)
 \end{align*}
 where $\lambda$ corresponds to the usual LASSO penalty hyperparameter and we have assumed $\sigma^2=1$ is known. One can then show the conditional posterior is
 \begin{align*}
     \tau^{(j)} \mid \beta^{(j)}\sim \text{Inverse-Gaussian}\left(\frac{\lambda}{|\beta^{(j)}|}, \lambda^2\right)
 \end{align*}
 which has posterior mean ${\lambda}/{|\beta^{(j)}|}$, thus matching the entry in Table \ref{tab:D_n_summary}. In the case where $|\beta^{(j)}|$ is close or equal to 0, another option for the plug-in estimate is the reciprocal form
 $$
    \frac{1}{
    E\!\left[
    1/\tau^{(j)}
    \mid
    \beta^{(j)}
    \right]
    }
    =
    \frac{\lambda^2}{
    \lambda\, |\beta^{(j)}| + 1
    }
$$
which can be shown from properties of the inverse-Gaussian distribution. The above is stable for small values of $|\beta^{(j)}|$.
 
\end{example}

\begin{example}[Gaussian continuous spike-and-slab]
Another popular class of priors is the continuous spike-and-slab prior \citep{george1993variable}, where 
\begin{align*}
   \tau^{(j)} &= \left(1-z^{(j)}\right) v_0^{-1} + z^{(j)} v_1^{-1}, \quad z^{(j)} \sim \text{Bernoulli}\left(\theta\right),
\end{align*}
where $z^{(j)} \in \{0,1\}$, $\theta \in (0,1)$ and $v_0 \ll v_1$. This corresponds to a Gaussian spike with variance $v_0$ and a slab with variance $v_1$. It is not too hard to show that
$$E\left(\tau^{(j)} \mid \beta^{(j)}\right) = (1-\gamma^{(j)})\,v_0^{-1} + \gamma^{(j)}\,v_1^{-1}$$
where
\begin{align*}
    \gamma^{(j)} = E\left(z^{(j)} \mid \beta^{(j)}\right) =  \frac{\theta\,\mathcal{N}(\beta^{(j)};\, 0, v_1)}{\theta\,\mathcal{N}(\beta^{(j)};\, 0, v_1) 
      + (1-\theta)\,\mathcal{N}(\beta^{(j)};\,  0, v_0)}\cdot
\end{align*}
We note that $\gamma^{(j)}$ is essentially the conditional posterior inclusion probability, and the above computation is also used in the EMVS algorithm of \citet{rovckova2014emvs}. 
 \end{example}

\begin{example}[Spike-and-slab LASSO]
The continuous spike-and-slab prior with the Gaussian kernel typically requires post-hoc thresholding of posterior inclusion probabilities for variable selection \citep{bai2021spike}. In contrast, the spike-and-slab LASSO prior \citep{rockova2018spike} replaces the Gaussian kernels with Laplace kernels, refining the \(\ell_1\) penalty by allowing adaptive shrinkage tailored to the magnitude of each coefficient. 

In this setting, we have a slightly more complex prior, where
\begin{align*}
    \frac{1}{\tau^{(j)}} \mid z^{(j)} &\sim \text{Exp}\left(\frac{\lambda_{z^{(j)}}^2}{2}\right), \quad 
    \lambda_{z^{(j)}} =
\left(1-z^{(j)}\right)\,\lambda_0
+
z^{(j)}\,\lambda_1, 
\quad z^{(j)}\sim
\mathrm{Bernoulli}(\theta),
\end{align*}
where again $\theta \in (0,1)$, and $\lambda_0 \gg \lambda_1$. Now, we have
\begin{align*}
    \left(\tau^{(j)}
\mid
\beta^{(j)},
z^{(j)}\right)
\sim
\text{Inverse-Gaussian}
\left(
\frac{
\lambda_{z^{(j)}}
}{
|\beta^{(j)}|
},
\lambda_{z^{(j)}}^2
\right),
\end{align*}
therefore giving
\begin{align*}
  E\left(\tau^{(j)} \mid \beta^{(j)}\right)=  E\left\{E\left(\tau^{(j)} \mid \beta^{(j)},z^{(j)}\right)\mid \beta^{(j)} \right\} =   \left({1-{\gamma}^{(j)}}\right)
\frac{\lambda_0}
{
\left|
\beta^{(j)}
\right|
}
+
{\gamma}^{(j)}
\frac{\lambda_1}
{
\left|
\beta^{(j)}
\right|
}
\end{align*}
where the inclusion probability $\gamma^{(j)}$ can again be shown to be the form
\begin{align*}
 \gamma^{(j)} =   \frac{
\theta\,\lambda_1\,
\exp\left(
-\lambda_1
|
\beta^{(j)}
|
\right)
}
{
(1-\theta)\,\lambda_0\,
\exp\left(
-\lambda_0
|
\beta^{(j)}
|
\right)
+
\theta\,\lambda_1\,
\exp\left(
-\lambda_1
|
\beta^{(j)}
|
\right)
}\cdot 
\end{align*}
\end{example}

We make a final remark regarding the uncertainty quantification in our method, which is particularly highlighted in the spike-and-slab examples. For variable selection priors,  an additional layer of uncertainty arises from the uncertainty in the variable inclusion indicators $z^{(j)}$, which is not handled in our case as we plug in an estimated value of $E(z^{(j)} \mid \beta_n^{(j)})$. Thus, our martingale posterior framework focuses exclusively on local coefficient uncertainty, and extending the framework to quantify model selection uncertainty is an important and nontrivial future direction of research.

\subsection{Adaptive weights $w(\beta^\top x)$}\label{sec:adaptive_I}
As discussed in Section \ref{sec:gen_w}, we briefly consider one other variant of $\mathcal{I}_i$ and $\alpha_i(x)$ which allows for $w(\cdot)$ to depend on $\beta_i$. Suppose we now have
\begin{align*}
    \mathcal{I}_i = \sum_{j = 1}^i w(\beta_{j-1}^\top X_j) X_j X_j^\top + D, \quad   \alpha_i(x) = \left\{{1 + w\left(\beta_{i}^\top x\right) x^\top \mathcal{I}_i^{-1} x}\right\}^{1/2}.
\end{align*}
The main difference is that $\mathcal{I}_i$ is no longer evaluated at a single $\beta$, instead relying on the whole sequence $\beta_{1:i-1}$ in an online manner.

In order to attempt to show weak design invariance, one can verify that the above choice gives
\begin{align*}
    E[Z_i Z_i^\top \mid  \mathcal{F}_{i-1}, X_i] &= \alpha_{i-1}(X_i)^2 \cdot \frac{w(\beta_{i-1}^\top X_i) \mathcal{I}_{i-1}^{-1} X_i X_i^\top \mathcal{I}_{i-1}^{-1}}{\left\{1 + w(\beta_{i-1}^\top X_i)X_i^\top \mathcal{I}_{i-1}^{-1} X_i \right\}^2} = \mathcal{I}_{i-1}^{-1} - \mathcal{I}_i^{-1},
\end{align*}
similar to Lemma \ref{lem_increment_cov}. We can then still show that
\begin{align*}
E\!\left\{
\left(\beta_N-\beta_n\right)\,\left(\beta_N-\beta_n\right)^\top
\mid
\mathcal F_n
\right\}
&= \mathcal{I}_n^{-1} - E\left(\mathcal{I}_N^{-1}\mid \mathcal{F}_n\right).
\end{align*} 
We now face the crux of the issue: we can no longer rely on a strong law of large numbers argument to show $E\left(\mathcal{I}_N^{-1}\mid \mathcal{F}_n\right)\to 0$, as $\mathcal{I}_N$ depends on the random $\beta_N$. We suspect this will introduce the need for stringent conditions on $w(t)$ and $P_X$, and will be much less general than our Theorem \ref{thm:predictive_invariance_mp}. Proving Theorem \ref{thm:as_invar} will also be significantly more challenging, as we also rely on the strong law of large numbers argument here. This approach will also significantly increase the computational cost, as the dependence of $\mathcal{I}_i$ on $\beta_{n+1:i}$ prevents precomputation. 
Finally, we note that having an adaptive weight will not have much practical impact, since the target limiting covariance would still be \(\mathcal I_n^{-1}\) if we can establish weak design invariance.

\section{Additional experiments and details}\label{sec:experiments}
\subsection{Computation}\label{sec:computation}

This subsection outlines additional computational details on the implementation of (\ref{eq:recur_update}). In particular, we discuss the computation of the initial estimate, the precomputation of terms involving $\mathcal{I}_i^{-1}$, and details on potential parallelization of predictive resampling. As outlined in the main paper, all methods are implemented in Julia \citep{bezanson2017julia} for computational expediency, due to the inherent recursive nature of our methodology.

Predictive resampling relies on initializing $\beta_n$ to the maximum a posteriori estimate corresponding to the appropriate prior. When the objective function is non-convex, such as for the continuous spike-and-slab prior, the numerical optimization procedure requires more care to obtain a good estimate. We can resort to standard techniques such as deterministic annealing \citep{ueda1998deterministic} or dynamic posterior exploration with a warm start which is often used in the spike-and-slab case \citep{rockova2018spike,bai2021spike}. For the simulation in Section \ref{sec:sim}, we implement an expectation-maximization algorithm as the base optimizer, which can be viewed as an extension of the EMVS algorithm of \citet{rovckova2014emvs} with an additional continuous latent variable due to the Student-$t$ likelihood. To carry out dynamic posterior exploration, we initialize at $\beta_n = 0_p$, then gradually decrease the spike variance along a ladder from $v_0 = 0.5$ to the target $v_0 =0.0005$ using the estimate from the previous $v_0$ value as the warm start for our expectation-maximization optimizer.

As discussed in Section \ref{sec:sim}, it is computationally convenient to reuse the same (random) sequence $X_{n+1:N}$ across predictive resampling trajectories, which can be justified theoretically as in Section \ref{sec:fixed_design}. In particular, the expensive terms in (\ref{eq:recur_update}) are the ones that involve $\mathcal{I}_i^{-1}$.
In the absence of parallelization as in Section \ref{sec:sim}, we utilize the Sherman-Morrison update to compute $(\mathcal{I}_i^{-1})_{(n + 1)\leq i \leq N}$ in $O\{(N-n)p^2\}$ operations given the sequence $X_{n+1:N}$ and initial $\mathcal{I}_n^{-1}$. In this process, we will have precomputed the terms $\mathcal{I}_i^{-1} X_i$ and $\{1 + w(\beta_n^\top X_i) X_i \mathcal{I}_{i-1}^{-1} X_i\}$ for $(n+1) \leq i \leq N$, which are exactly needed to compute $\mathcal{I}_i^{-1}s(\beta_{i-1},Y_i \mid X_i)$ and $\alpha_{i-1}(X_i)$ respectively when predictive resampling. We highlight that this precomputation is very similar to that of recursive least squares, so we may be able to further leverage techniques within this literature to speed up or approximate our computation, which we leave for future work. 

Finally, we briefly discuss the potential for parallelization. Unlike Gibbs sampling, predictive resampling is inherently parallelizable across trajectories. As a result, the run-time for predictive resampling (after precomputation) can be reduced to essentially the run-time of a single trajectory given sufficient parallelization. Parallelization may also be used to accelerate precomputation, e.g. through block processing of $\mathcal{I}_i^{-1}$ for blocks of $\{n+1,\ldots,N\}$, and also allowing for different $X_{n+1:N}$ sequences across predictive resampling trajectories if desired. 

\subsection{Impact of correction term $\alpha_i(x)$}
We empirically demonstrate the importance of the scalar correction factor $\alpha_i(x)$ in this subsection. Figure \ref{fig:pr_vs_bayes_sim_correction} illustrates the limit under the setting of Section \ref{sec:sim} using recursive update \eqref{eq:recur_update} \textit{without} correction:
\begin{align*}
    \beta_i= \beta_{i-1}
&+ \mathcal{I}_i^{-1}\,s(\beta_{i-1}, Y_i \mid X_i).
\end{align*}
We see that the marginal posteriors are clearly noticeably different for $P_X = \mathcal{N}(0, I_p)$ versus $\mathcal{N}(0,10^2 \, I_p)$, highlighting the importance of the scalar correction factor for (weak) design invariance. Without the scalar factor, we no longer have the telescoping argument required for Theorem \ref{thm:predictive_invariance_mp}.
\begin{figure}[h!]
\centering
\includegraphics[width=0.96\linewidth]{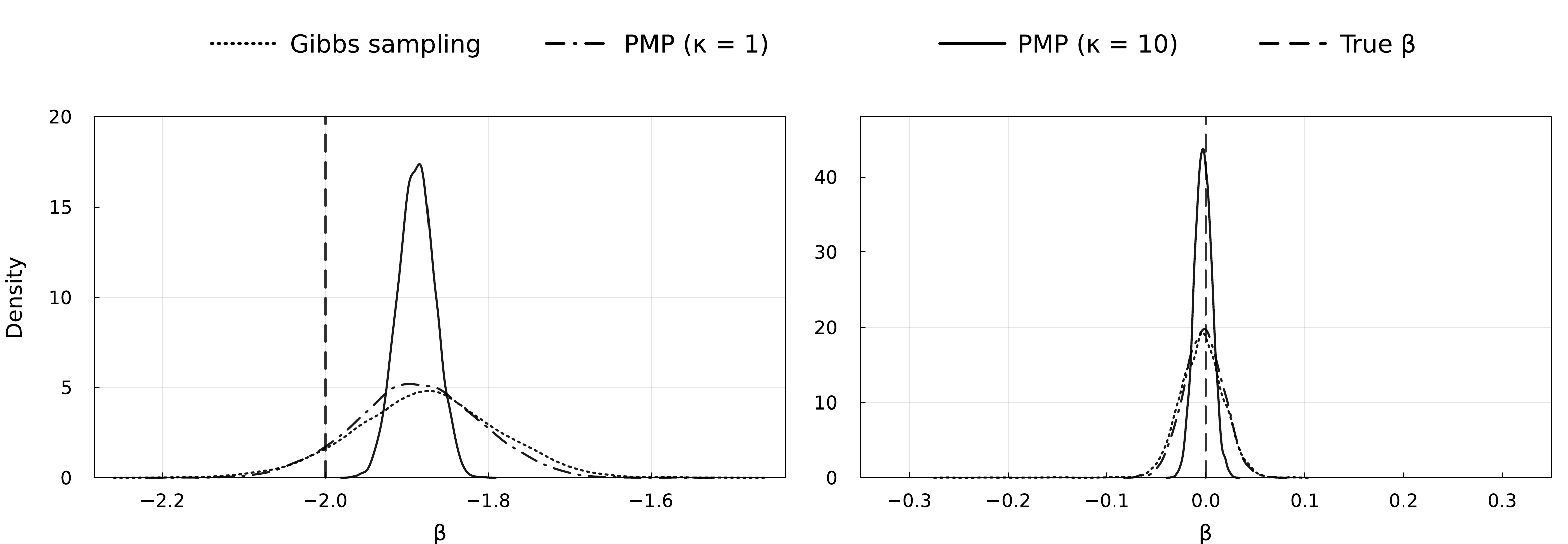}
\caption{Posteriors for the Student-$t$ example for selected (Left) active coefficient; (Right) inactive coefficient. Predictive resampling is with update rule \eqref{eq:recur_update} minus the correction term $\alpha_{i-1}(X_i)$, and using $P_X = \mathcal N(0, I_{p})$ versus $P_X = \mathcal N(0, 10^2\,I_{p})$. } 
\label{fig:pr_vs_bayes_sim_correction}
\end{figure}

\subsection{Convergence rate of predictive resampling}\label{sec:kappa}
In this subsection, we explore the impact of the scale parameter in $P_X = \mathcal{N}(0,\kappa^2 I_p)$ on the convergence rate of predictive resampling under the setting of Section \ref{sec:sim}. Figure \ref{fig:pr_convergence_robust} illustrates predictive resampling trajectories under different values of $\kappa$, where it is clear that increasing $\kappa$ allows the trajectories to converge to their limiting values faster, with diminishing returns after $\kappa = 10$. We also highlight a notable observation: predictive resampling will always require a truncation value of at least $N \geq p+1$ for convergence, which is not too surprising as otherwise the parameter will be non-identifiable given $(X_{1:N},Y_{1:N})$. As discussed in the main paper, a practical recommendation to avoid numerical instability is to increase $\kappa$ until the predictive resampling trajectories converge sufficiently quickly as desired; we see in Figure \ref{fig:pr_convergence_robust} that $\kappa = 10$ is sufficient. 
\begin{figure}[h!]
\centering
\includegraphics[width=0.98\linewidth]{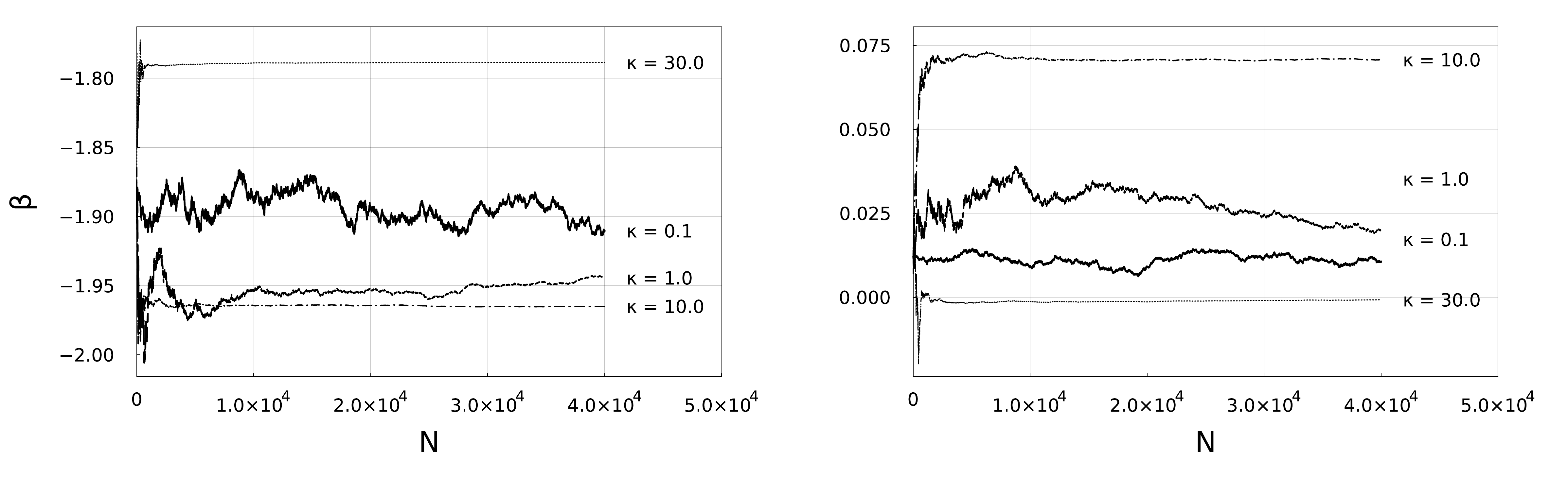}
\caption{Trajectory plots for the Student-$t$ example for selected (Left) active coefficient; (Right) inactive coefficient. Predictive resampling is with update rule \eqref{eq:recur_update} using $P_X = \mathcal N(0, \kappa^2 I_{p})$ where $\kappa = (0.1, 1.0, 10.0, 30.0)$. } 
\label{fig:pr_convergence_robust}
\end{figure}

To understand why increasing $\kappa$ can increase the convergence rate of predictive resampling, we look at the remaining variance at truncation $N$. From Theorem \ref{thm:predictive_invariance_mp}, we have
   \begin{align*}
E\!\left\{\left(\beta_\infty - \beta_N\right)\left(\beta_\infty - \beta_N\right)^\top \right\} = E\left(\mathcal{I}_N^{-1}\right)
\end{align*}
where the expectation is over $X_{n+1:N} \iid P_X$. Increasing the convergence speed of predictive resampling is thus equivalent to reducing the size of $E\left(\mathcal{I}_N^{-1}\right)$, which is connected to the amount of information in each imputed sample $(Y_i,X_i)$. For the Student-$t$ example in Section \ref{sec:sim}, the above variance term can be reduced by increasing $\kappa$ as $\mathcal{I}_N$ is approximately proportional to $\sum_{i = 1}^N X_i X_i^\top$. \vspace{5mm}

\newpage

\subsection{Main simulation and additional examples}\label{sec:addit_examples}

We include additional details on the experiments in the main paper, and illustrate two additional examples to demonstrate the weak design invariance beyond the location model setting. As in the main paper, we generate $B = 5000$ posterior samples (or effective sample size for Markov chain Monte Carlo) for all examples. The observations are again generated from the corresponding parametric model in each experiment.

For the Student-$t$ example in Section \ref{sec:sim}, we note that the precomputation required 0.04s out of the total run-time of 0.12s for $P_X = \mathcal{N}(0, 10^2\, I_p)$. For the case where $P_X = \mathcal{N}(0, I_p)$, we found that $N = 20000$ was sufficient for truncation, where we required 1.31s for precomputation and 2.63s for predictive resampling, resulting in a run-time of 3.94s.

Figure \ref{fig:pr_vs_bayes_sim_gamma} illustrates Gamma regression with a fixed shape of $\alpha =1$, a log link, and a ridge prior $\mathcal{N}(0, {I}_p)$, where $n=p=10$ with 5 active coordinates of magnitude $\approx 0.45$.  For the Gamma family, we have $w(t) = \alpha$, so we can rely on Assumption \ref{as:lyapunov_easy}. We see that the martingale posteriors appear identical for $P_X = \mathcal{N}(0,I_p)$ and $\mathcal{N}(0, 2^2 \, I_p)$ due to weak design invariance. Furthermore, we highlight that this design invariance holds even under clear non-Gaussianity of the martingale posterior, indicating the strength of weak design invariance.
For the parametric martingale posterior, the initial maximum a posteriori estimate $\beta_n$ is computed using the L-BFGS optimizer \citep{liu1989limited}. For predictive resampling, $N = 1000$ is sufficient for convergence, with a precomputation time of 0.35\,ms followed by a run-time of 22\,ms for $B = 5000$ samples. We implement the traditional Bayesian posterior as a comparator using the No-U-Turn sampler \citep{hoffman2014no} in the \texttt{Turing.jl} package \citep{fjelde2025turing} as the prior is not partially conjugate. Here we used 4000 adaptation steps, followed by 8000 samples to attain an effective sample size of $B = 5000$, which required 1.10\,s of run-time.

\begin{figure}[!h]
\centering
\includegraphics[width=0.96\linewidth]{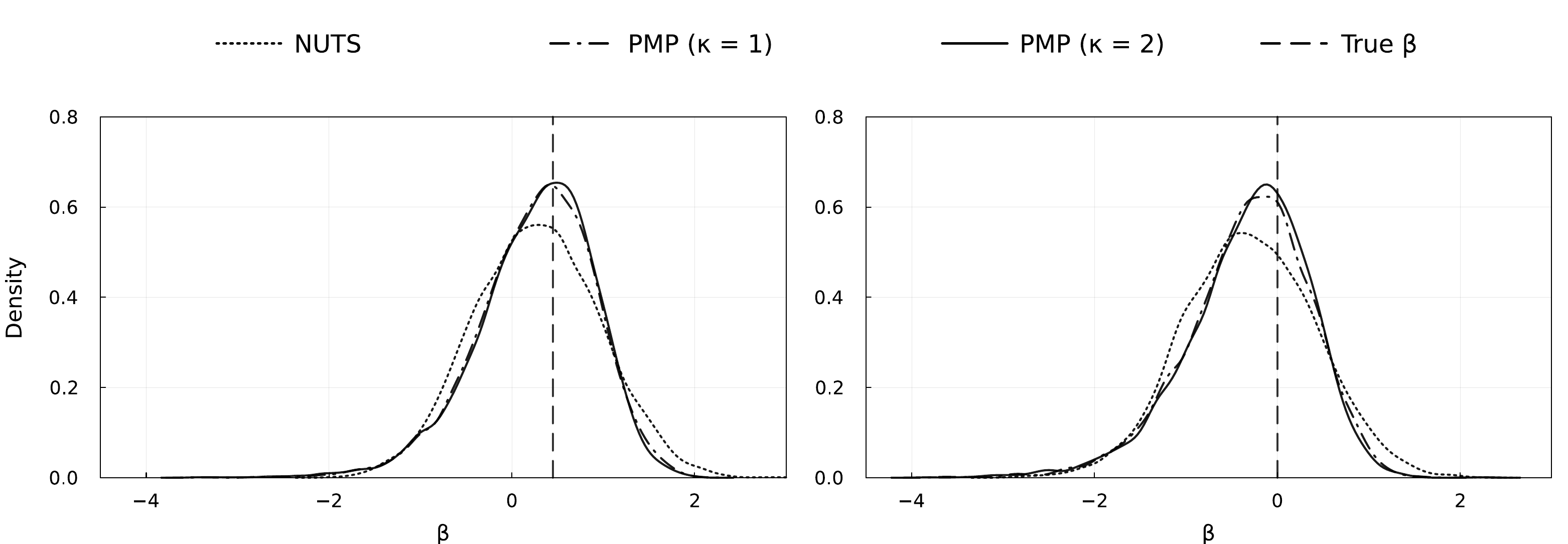}
\caption{Posteriors for the gamma example for selected (Left) active coefficient; (Right) inactive coefficient. Predictive resampling is with update rule \eqref{eq:recur_update} using $P_X = \mathcal N(0, I_{p})$ versus $P_X = \mathcal N(0, 2^2\,I_{p})$. } 
\label{fig:pr_vs_bayes_sim_gamma}
\end{figure}

Figure \ref{fig:pr_vs_bayes_sim_logistic} illustrates a logistic regression example using a ridge prior $\mathcal{N}(0, 0.01{I}_p)$, where $n=1000, p=500$ with 5 active coordinates of magnitude $\approx 0.45$. This is an example where $w(t)$ is not constant; see Example \ref{ex:logistic}. Here we choose two options for $P_X$ with compact support $[-6,6]^p$, in connection to Assumption \ref{as:lyapunov_hard}. All coordinates  $X^{(j)}$ are drawn independently from the uniform distribution or the truncated normal with $\sigma = 3$ on $[-6,6]$ in the two cases respectively. Again, weak design invariance is clearly preserving robustness of the martingale posterior to the choice of $P_X$ in this setting. The initial maximum a posteriori estimate $\beta_n$ is computed using \texttt{glmnet}. For predictive resampling, we find that $N = 5000$ is sufficient for convergence, which requires a precomputation time of 0.32\,s followed by a run-time of 4.83\,s to generate $B = 5000$ samples. As a comparison, we implement the traditional Bayesian posterior with the Gibbs sampler of \cite{polson2013bayesian}. Here we generate 6000 burn-in samples followed by another 6000 samples to achieve a mean effective sample size of $B = 5000$ in the active coordinates, requiring a total run-time of 151.64\,s. We also highlight that a single coordinate of the maximum a posteriori estimate $\beta_n^{(j)}$ does not necessarily align with the mode of the marginal posterior, which explains the observation in Figure \ref{fig:pr_vs_bayes_sim_logistic}.

\begin{figure}[!h]
\centering
\includegraphics[width=0.96\linewidth]{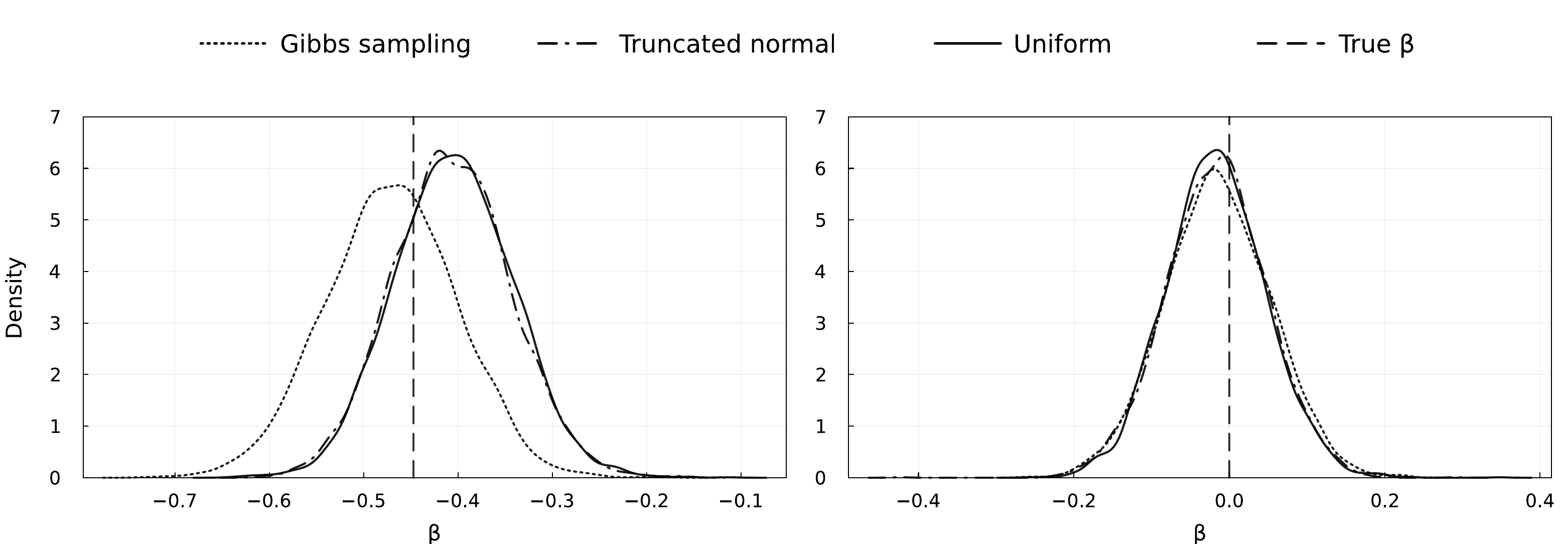}
\caption{Posteriors for the logistic example for selected (Left) active coefficient; (Right) inactive coefficient. Predictive resampling is with update rule \eqref{eq:recur_update} using $P_X = \text{TruncNorm}(-6,6; \sigma = 3)^p$ versus $P_X = \text{Uniform}(-6,6)^p$. } 
\label{fig:pr_vs_bayes_sim_logistic}
\end{figure}

\end{appendices}

\end{document}